\documentclass{ws-rmp1}

\usepackage[english]{babel}
\usepackage[initials,nobysame]{amsrefs}
\newcommand{\N}{\mathbb{N}}
\newcommand{\R}{\mathbb{R}}
\newcommand{\C}{\mathbb{C}}
\newcommand{\Rd}{\mathbb{R}^d}
\newcommand{\Rdn}{\mathbb{R}^{dn}}
\newcommand{\RD}{\mathbb{R}^D}
\newcommand{\ZD}{\mathbb{Z}^D}
\newcommand{\Zd}{\mathbb{Z}^d}
\newcommand{\Zdn}{\mathbb{Z}^{dn}}

\newcommand{\dE}{\mathrm{d}E}

\newcommand{\abf}{\mathbf{a}}
\newcommand{\bbf}{\mathbf{b}}
\newcommand{\ubf}{\mathbf{u}}
\newcommand{\vbf}{\mathbf{v}}
\newcommand{\wbf}{\mathbf{w}}
\newcommand{\xbf}{\mathbf{x}}
\newcommand{\ybf}{\mathbf{y}}

\newcommand{\id}{\boldsymbol{1}}
\newcommand{\Prob}{\mathbb{P}}
\newcommand{\Ex}{\mathbb{E}}

\newcommand{\eps}{\varepsilon}

\newcommand{\calF}{\mathcal{F}}
\newcommand{\calI}{\mathcal{I}}

\newcommand{\Hhat}{\widehat{H}}
\DeclareMathOperator{\dist}{dist}
\DeclareMathOperator{\diam}{diam}
\DeclareMathOperator{\md}{d}
\DeclareMathOperator{\Tr}{Tr}
\DeclareMathOperator{\supp}{supp}
\DeclareMathOperator{\Imag}{Im}
\DeclareMathOperator{\Real}{Re}

\begin{document}

\markboth{Fauser, Warzel}
{Multiparticle localization for disordered systems}

%
\catchline{}{}{}{}{}
%

\title{Multiparticle localization for disordered systems on continuous space via the fractional moment method}

\author{MICHAEL FAUSER}

\address{Zentrum Mathematik, TU M\"unchen, Boltzmannstrasse 3, 85747 Garching, Germany\\ 
\email{fauser@ma.tum.de} }

\author{SIMONE WARZEL}

\address{Zentrum Mathematik, TU M\"unchen, Boltzmannstrasse 3, 85747 Garching, Germany\\ 
\email{warzel@ma.tum.de} }

\maketitle


\begin{abstract}
We investigate spectral and dynamical localization of a quantum system of $ n $ particles on $ \R^d $ which are subject to a random potential and interact through a pair potential which may have infinite range. We establish two conditions which ensure spectral and dynamical localization near the bottom of the spectrum of the $ n $-particle system: i)~localization is established in the regime of weak interactions supposing one-particle localization, and ii)~localization is also established under a Lifshitz-tail type condition on the sparsity of the spectrum. In case of polynomially decaying interactions, we provide an upper bound on the number of particles up to which these conditions apply.

\end{abstract}



\section{Introduction}

Rigorous mathematical analysis of random operators has been a vital field of mathematical physics for the last few decades. For one-particle systems localization in the spectral and dynamical sense has been established rigorously in various regimes of energy and disorder both for lattice and continuum models \cites{GMP,KS,FS,AM,BK,Sto01,GK13}.
Despite of its physical relevance, the question of localization in disordered many-particle systems, however, is much less understood. So far, the mathematical analysis of interacting disordered systems with a 
macroscopic number of particles has been mostly restricted to questions of the existence and properties of thermodynamic quantities \cites{Ven,BL,KM11} -- with the exception of a proof of Bose-Einstein condensation for the disordered Lieb-Liniger model in one dimension \cite{SYZ12}. 

In the course of the last years, several  localization results have been obtained for systems of a fixed number of interacting particles. Persistence of localization has first been established for lattice systems, using multiscale \cites{CS09-1,CS09-2} or fractional moment~\cite{AW09} techniques. The multiscale analysis has been adapted to continuum multiparticle systems \cite{CBS11}. In addition, an improved version of the multiscale analysis, commonly referred to as ``bootstrap multiscale analysis'' has recently been extended to multiparticle systems as well, both for lattice \cite{KN13-1} and continuum models \cite{KN13-2}.

The aim of the present work is to complement these results and present an adaption of the multiparticle fractional-moment method developed in \cite{AW09} to the continuum. The corresponding multiparticle random Schr\"odinger operator has an alloy-type external potential and includes an interaction potential with rapidly decaying two-body interaction terms. We focus on localization in energy regimes at the bottom of the spectrum and pay particular attention to how the strength and the decay of the interaction between different particles affects the localization estimates.

Our main results  are divided into two different types of decay of the interaction. In the first case, we assume that the interaction between particles located at points $x$ and $y$ decays (sub)exponentially, i.e., as $e^{-\mu |x-y|^\zeta}$ with $ \mu > 0 $ and $0<\zeta\leq1$. This yields strong dynamical localization with decay estimates of the same order, albeit with respect to the ``Hausdorff distance'' $\dist_H$ and not the regular Euclidean metric on the configuration space of the particles. In the second case, we only assume fast polynomial decay which allows us to establish localization for systems consisting of at most a maximal number of particles. Although the interaction decays only polynomially, the localization estimates decay as $e^{-\mu \dist_H(\xbf,\ybf)^\zeta}$, with $ \mu > 0 $ and $\zeta\in(0,1]$ getting smaller as the particle number increases.

In our analysis it is essential to switch between different notions of localization, namely fractional moment localization and eigenfunction correlator localization. Section \ref{section:fmec} discusses these two terms and their relation in a fairly general way that is of interest in itself and might be of use in other settings as well. The main goal in that section is a proof that fractional moment localization and eigenfunction correlator localization are equivalent up to the localization length and a small change in the interval of validity.

The main part of our analysis is placed in Section \ref{section:mploc}. We prove the inductive step that localization for up to $n-1$ particles implies localization for $n$ particles. We proceed in several steps, first considering partially interactive systems and then using this to prove a ``rescaling inequality'' for the fully interacting systems. Sufficiently good estimates on an initial length scale will then yield localization for the $n$-particle system. Finally, in Section \ref{section:proofofmainresults}, we will use these results to prove the main theorems that are presented in Section \ref{section:mainresults}.

\subsection{The model}\label{section:model}

We start with a one-particle Hamiltonian
\begin{equation}
H^{(1)}(\omega) = H_0^{(1)} + V^{(1)}(\omega),
\end{equation}
acting in $L^2(\Rd)$, which is the sum of a deterministic term
$ H_0^{(1)} = -\Delta_d + V_0 $, 
consisting of the $d$-dimensional Laplacian and a background potential $V_0$, and an alloy-type random potential
\begin{equation}
V^{(1)}(\omega,x) = \sum_{\zeta\in\Zd} \eta_\zeta(\omega)U(x-\zeta).
\end{equation}
The  $n$-particle operator $H^{(n)}$, acting on the tensor product space $(L^2(\Rd))^n=L^2(\Rdn)$, has the form
\begin{equation}
H^{(n)} = \sum_{j=1}^n (H_0^{(1)}(j) + V^{(1)}(j)) + \alpha_W W^{(n)} = H_0^{(n)} + V^{(n)} + \alpha_W W^{(n)},
\end{equation}
where $H_0^{(1)}(j)=\id\otimes\cdots\otimes H_0^{(1)}\otimes\cdots\otimes\id$ with $H_0^{(1)}$ acting on the variable $x_j$. The same holds for $V^{(1)}(j)$. Denoting the $n$-particle configuration by $\xbf=(x_1,\ldots,x_n)\in\Rdn$, the random potential reads 
\begin{equation} V^{(n)}(\omega,\xbf)  = \sum_{\zeta \in \mathbb{Z}^d} \eta_\zeta(\omega) N_\zeta(\xbf) \quad \mbox{with}\quad  N_\zeta(\xbf) = \sum_{j=1}^n U(x_j - \zeta) \, .
\end{equation} 
The additional operator $W^{(n)}$ denotes the potential corresponding to interactions between different particles. Its strength is controlled by $\alpha_W\geq 0$. 
We make the following assumptions:
\begin{enumerate}
  \item[(I)] The background potential $V_0:\Rd\rightarrow\R$ is bounded and $\Zd$-periodic.
  \item[(II)] The single site potential $U:\Rd\rightarrow\R$ is non-negative, bounded and supported in a ball of radius $r_U>0$ around $0$. Furthermore, the added single-site potentials satisfy the ``covering-condition''
  \begin{equation}
  \inf_{x\in\Rd}\sum_{\zeta\in\Zd}U(x-\zeta) > 0.
  \end{equation}
  \item[(III)] The random variables $\eta_\zeta$, $\zeta\in\Zd$, are independent and identically distributed with an absolutely continuous marginal distribution. The corresponding density $\rho$ is bounded and has compact suppport. Without loss of generality, $\inf\supp\rho=0$.
  \item[(IV)] The interaction potential $W^{(n)}$ has the form
  \begin{equation}
  W^{(n)}(\xbf) = \sum_{j<k}w(x_j-x_k) \, . 
  \end{equation}
A bound on the two-particle interaction potential $w$ is given by $|w(x_j-x_k)|\leq w_b(|x_j-x_k|)$ with a monotonely decreasing function $w_b$ satisfying $w_b(r)\leq C(1+r)^{-1}$ for some $C< \infty $.
\end{enumerate}

As the variables $\eta_\zeta$, $\zeta\in\Zd$, are independent and identically distributed, the Hamiltonians $H^{(n)}(\omega)$ form an ergodic family of operators with respect to the translations $(x_1,\ldots,x_n)\mapsto(x_1+a,\ldots,x_n+a)$, $a\in\Zd$. The non-randomness of the spectrum is a standard consequence (cf. \cites{Ki89,CaLa90}).
\begin{proposition}
The spectrum of $H^{(n)}$ is almost surely a non-random set $ \Sigma^{(n)} $.
\end{proposition} 
We denote the almost sure infimum of the spectrum by
\begin{equation}
E_0^{(n)}:=\inf\Sigma^{(n)}.
\end{equation}
For later purpose, we also note the following:
\begin{proposition}\label{prop:gs}
If $j_1+\ldots+j_l=n$, then
  \begin{equation}
  E_0^{(n)}\leq \sum_{i=1}^l E_0^{(j_i)}.
  \end{equation}
\end{proposition}
 In the case of a repulsive interaction potential $w\geq 0$, this is in fact an equality. We give an easily generalizable sketch of the proof of the inequality $E_0^{(2)}\leq E_0^{(1)}+E_0^{(1)}$:
Choose $\phi^{(1)}\in C_c^\infty(\Rd)$ such $\langle\phi^{(1)},H^{(1)}\phi^{(1)}\rangle\approx E_0^{(1)}$ and $\|\phi^{(1)}\|=1$. Now let $\psi^{(1)}$ be a spatially translated version of $\phi^{(1)}$ with a support that is far away from the support of $\phi^{(1)}$, so that $\langle(\psi^{(1)}\otimes\phi^{(1)}),W^{(2)}(\psi^{(1)}\otimes\phi^{(1)})\rangle\approx0$. Due to the assumptions on $V^{(1)}$, one can almost surely choose this translation in such a way that $\langle\psi^{(1)},V^{(1)}\psi^{(1)}\rangle\approx\langle\phi^{(1)}V^{(1)}\phi^{(1)}\rangle$. Consequently, $\langle(\psi^{(1)}\otimes\phi^{(1)}),H^{(2)}(\psi^{(1)}\otimes\phi^{(1)})\rangle\approx 2E_0^{(1)}$ with probability one, which implies the assertion.\\

The results we prove are expressed in estimates on the spatial decay of certain quantities. This decay is measured with respect to the pseudo-metric
\begin{equation}
\dist_H(\xbf,\ybf) = \max\big\{\max_j\min_k|x_j-y_k|,\max_k\min_j|x_j-y_k|\big\}
\end{equation}
for configurations $\xbf,\ybf\in\Rdn$, which is in fact the Hausdorff distance of the sets $\{x_1,\ldots,x_n\}$ and $\{y_1,\ldots,y_n\}$ in $\Rd$.
The norm $|\cdot|$ we use on $\Rd$ and on $\Rdn$ is the $\max$-norm. Balls around a point $x\in\Rd$ with a radius $r>0$ are denoted by $\Lambda_r(x):=\{y\in\Rd\,|\,|y-x|<r\}$, whereas balls around configurations $\xbf\in\Rdn$ are denoted by $B_r(\xbf):=\{\ybf\in\Rdn\,|\,|\ybf-\xbf|<r\}$. The characteristic function $\id_{B_{1/2}(\xbf)}$ of $B_{1/2}(\xbf)$ is denoted by $\chi_\xbf$.

Most of the analysis in the following sections deals with restrictions of $H^{(n)}$ to finite volumes. We will always impose Dirichlet boundary conditions, which will be denoted by $H^{(n)}_{\Omega^n}$ if the Hamiltonian is restricted to $\Omega^n$.

\subsection{Main results}\label{section:mainresults}

The central tool in our analysis is the concept of fractional moment localization, which we define as follows:

\begin{definition}
A bounded interval $I$ is a regime of $n$-particle fractional moment localization of order $\gamma\in(0,1]$ if there exist $C,\mu\in (0,\infty)$ and $s\in(0,1)$ such that
\begin{equation}
\sup_{\substack{\Omega\subset\Rd \\ \mathrm{open,}\\ \mathrm{bounded}}}\sup_{\substack{\Real z\in I \\ 0<|\Imag z|<1}}\Ex\big[\|\chi_\xbf(H_{\Omega^n}^{(n)}-z)^{-1}\chi_\ybf\|^s\big] \leq Ce^{-\mu\dist_H(\xbf,\ybf)^\gamma}
\end{equation}
for all $\xbf,\ybf\in\Rdn$.
\end{definition}
As we will see at the end of Section \ref{section:fmec}, this implies dynamical localization in $I$ in the sense that
\begin{equation}
\Ex\big[\sup_{t\in\R}\|\chi_\xbf e^{-itH^{(n)}}P_I(H^{(n)})\chi_\ybf\|\big] \leq Ae^{-\nu\dist_H(\xbf,\ybf)^\gamma}
\end{equation}
for some $A,\nu>0$ and all $\xbf,\ybf\in\Rdn$, where we denote by $P_I(H^{(n)})$ the spectral projection of $H^{(n)}$ to the interval $I$. In particular this implies absence of continuous spectrum in $I$ and  (sub)exponential decay of the eigenfunctions with eigenvalues in~$I$, cf. Theorem \ref{thm:locinfvolume}.

In the following, we state our main results. They are divided into two different cases depending on the decay of the two-particle interaction. 

\subsubsection*{Case 1: (Sub)exponentially decaying interaction}

We provide criteria for fractional-moment localization near the bottom of the almost-sure spectrum in the regime of weak interaction and  Lifshitz tails.
\begin{theorem}\label{thm:main:case1}
Assume that $H^{(n)}$ satisfies \textnormal{(I)-(IV)} and that $w_b(r)\leq c_we^{-\mu_w r^\gamma_w}$ for some $c_w,\mu_w>0$, $\gamma_w\in(0,1]$ and all $r\geq0$.
\begin{enumerate}
  \item[\textnormal{(i)}] Let $I^{(1)}=[E_0^{(1)},E_0^{(1)}+\eta^{(1)}]$ be a bounded interval of length $\eta^{(1)}>0$ that is a regime of one-particle fractional moment localization of order $\gamma_w\in(0,1]$.
For all $\eta^{(n)}\in(0,\eta^{(1)})$ there is $\alpha_0^{(n)}>0$ such that for any $\alpha_W\in [0,\alpha_0^{(n)}]$ the interval $I^{(n)}=[E_0^{(n)},E_0^{(n)}+\eta^{(n)}]$ is a regime of $n$-particle fractional moment localization of order $\gamma_w$.
  
  \item[\textnormal{(ii)}] Suppose that for any $k\in\{1,\ldots,n\}$ there exist $\xi^{(k)}>22dk$ and $L^{(k)}>0$ such that
\begin{equation}\label{eq:Lifshitz}
\Prob\big(E_0(H_{B_L(\xbf)}^{(k)}) \leq E_0^{(n)}+ L^{-1}\big) \leq L^{-\xi^{(k)}} 
\end{equation}
for all $\xbf\in\R^{dk}$ and $L\geq L^{(k)}$. Then there is $\eta^{(n)}>0$ such that $I^{(n)}=[E_0^{(n)},E_0^{(n)}+\eta^{(n)}]$ is a regime of $n$-particle fractional moment localization of order $\gamma_w$.
\end{enumerate}
\end{theorem}

One-particle fractional moment localization, which constitutes the assumption of the first part of the theorem, has been pioneered in \cites{AENSS}. In particular,  the authors show  under an additional assumption on the single-site distribution that any bounded interval $ I^{(1)} $ is a regime of one-particle fractional moment localization of order $1$ provided the disorder strength is sufficiently large. They also provide criteria on the density of states which imply this token. This criterion has been refined and simplified for the bottom of the spectrum in \cite{BNSS} . 
In one dimension $ d = 1 $, one-particle fractional moment localization up to arbitrary energies has been established in \cite{HSS10}. 

In the multi-dimensional case $ d > 1 $, even if one-particle localization is established at higher energies (e.g. at band edges), one cannot generally expect multi-particle localization in energy regimes separated from the bottom the spectrum. This is due to the possible appearance of extended states in the one-particle model at lower energies. 

As for the assumption in the second part of the theorem, finite-volume estimates of the form~\eqref{eq:Lifshitz}
are well-known in the case $n=1$, as they appear in the analysis of Lifshitz tails, cf.~\cite{Sto01}, and are used in the analysis in \cites{AENSS,BNSS}. In an $n$-particle system with repulsive interaction ($w\geq0$), the estimate follows directly from the one-particle case due to the fact that $E_0(H_{B_L(\xbf)}^{(n)})\geq \sum_j E_0(H_{\Lambda_L(x_j)}^{(1)})$ and $E_0^{(n)}=nE_0^{(1)}$. (This fact has already been used in \cites{Ekanga1,KN13-2}.) However, if the interaction is (partially) repulsive, a proof becomes more involved. One way to achieve this is to employ the techniques that are used in the proof of the inequality in the one-particle case, which is possible at least under certain assumptions on the interaction potential.

\subsubsection*{Case 2: Polynomially decaying interaction}

The analogue of Theorem~\ref{thm:main:case1} in case of polynomial decay is
\begin{theorem}\label{thm:main:case2}
Assume that $H^{(n)}$ satisfies \textnormal{(I)-(IV)} and that $w_b(r)\leq c_wr^{-p_w}$ for some $c_w,p_w>0$ and all $r\geq0$.
\begin{enumerate}
  \item[\textnormal{(i)}] Let $n<(p_w+8d)/(48d)$ and $I^{(1)}=[E_0^{(1)},E_0^{(1)}+\eta^{(1)}]$ be a bounded interval of length $\eta^{(1)}>0$ that is a regime of one-particle fractional moment localization of order $\beta^{(1)}\in(0,1]$. 
For all $\eta^{(n)}\in(0,\eta^{(1)})$ there is $\alpha_0^{(n)}>0$ auch that, if $\alpha_W\in [0,\alpha_0^{(n)}]$, then $I^{(n)}=[E_0^{(n)},E_0^{(n)}+\eta^{(n)}]$ is a regime of $n$-particle fractional moment localization of order $\beta^{(n)}=\beta^{(1)}/(1+(n-1)\beta^{(1)})$.
  \item[\textnormal{(ii)}] Suppose $n<(p_w+8d)/(48d)$ and that for any $k\in\{1,\ldots,n\}$ there exist $\xi^{(k)}>22dk$ and $L^{(k)}>0$ such that
\begin{equation}
\Prob\big(E_0(H_{B_L(\xbf)}^{(k)}) \leq E_0^{(n)}+ L^{-1}\big) \leq L^{-\xi^{(k)}} 
\end{equation}
for all $\xbf\in\R^{dk}$ and $L\geq L^{(k)}$. Then there exists $\eta^{(n)}\in(0,\eta^{(1)})$ such that $I^{(n)}=[E_0^{(n)},E_0^{(n)}+\eta^{(n)}]$ is a regime of $n$-particle fractional moment localization of order $\beta^{(n)}=\beta^{(1)}/(1+(n-1)\beta^{(1)})$.
\end{enumerate}
\end{theorem}

In contrast to Theorem \ref{thm:main:case1}, this theorem yields a bound on fractional moments that decays significantly faster than the interaction potential. The price we pay, however, is that we can prove this only up to a maximal number of particles and that the decay becomes slower as $n$ increases. The theorem does not exclude the possibility that there is no localization for large numbers of particles.\\

\subsection{Comparison with existing results}
Most results on localization for multi-particle systems deal with  lattice models. Aside from the pioneering works~\cites{AW09,CS09-1,CS09-2}, 
progress has been made in the lattice case on the energy regime in which localization is established and the regularity of the distribution of the random variables \cites{Ekanga1,Ekanga11,Ekanga2}, see also the recent monograph \cite{CS14} and references therein. As in the one-particle case, the multi-scale approach~\cites{CS09-1,CS09-2,CS14} has the advantage of being able to accommodate more singular distributions of the random variables whereas the fractional moment method \cite{AW09} allows one to establish exponential dynamical localization and not only polynomial decay or sub-exponential decay. The latter is obtained through the bootstrap multi-scale method which has only recently been generalised to the multi-particle setup \cite{KN13-1}.

For systems on continuous space, only two results are available so far. The paper~\cite{CBS11} uses the multi-scale method to establish polynomial dynamical localization for a multi-particle model with alloy-type random potential similar to the one studied in this paper. 
Aside from the difference in the decay rate, they establish the analogue of Theorem~\ref{thm:main:case1}(ii), i.e., localization in a sufficiently small interval above the bottom of the spectrum. 
As mentioned above, their method allows them to treat more singular distributions of the random variables. In the recent work~\cite{KN13-2} 
the polynomial decay has been improved to sub-exponential decay using the bootstrap method. 

In contrast to \cites{CBS11,KN13-2}, which allow for non-negative interactions of finite-range only, we investigate interactions which may have infinite range. For lattice models, (complete) localization in the presence of subexponentially decaying interactions has been considered 
before~\cite{Ch12}. Our novel point here is to study also the case of (fast) polynomial decay. 

\section{Fractional moment and eigenfunction correlator localization}\label{section:fmec}

The focus of this section is the notion of \textit{fractional moment localization} and its relation to other forms of localization. The main results in Theorems \ref{thm:main:case1} and \ref{thm:main:case2} are bounds on the spatial decay of fractional moments of the resolvent. It is well known from the one-particle case  \cite{AENSS} that estimates of this type imply pure point spectrum with exponentially decaying eigenfunctions as well as a strong form of dynamical localization. These results have been proved to carry over to multiparticle systems on the lattice in \cite{AW09}. An important term in this context are eigenfunctions correlators, which will be used in our analysis as well. We define \textit{eigenfunction correlator localization} as follows:

\begin{definition}
An interval $I\subset\R$ is a regime of $n$-particle eigenfunction correlator localization of order $\gamma\in(0,1]$ if there exist $C,\mu>0$ such that
\begin{equation}
\Ex\bigg[\sum_{E\in\sigma(H_{\Omega^n}^{(n)})\cap I}\|\chi_\xbf P_E(H_{\Omega^n}^{(n)})\chi_\ybf\|\bigg] \leq Ce^{-\mu\dist_H(\xbf,\ybf)^\gamma}
\end{equation}
for all bounded, open sets $\Omega\subset\Rdn$.
\end{definition}
As we will see, fractional moment localization indeed implies eigenfunction correlator localization, a fact that is well-known from the one-particle theory. In the multiparticle setting, however, the right choice of the metric is crucial \cite{AW09}. We will address this problem in our discussion below. An essential element in the multiparticle analysis of \cite{AW09}  was that exponential decay of the eigenfunction correlator implies exponential decay of fractional moments as well, so that the two terms are in a sense equivalent. In contrast to their analysis, we do not consider complete localization, but localization in a bounded interval I. As it turns out, we cannot prove that (sub)exponential decay of the eigenfunction correlator in a bounded interval implies (sub)exponential decay of fractional moments of the resolvent with energies in the interval $I$. Instead, one needs to restrict the energies to a slightly smaller interval $J\subset I$ that does not extend to the endpoints of $I$.\\

The analysis of different notions of localization and their relation to each other is not an issue specific to the fractional moment method, see \cite{GT12} 
and references therein.
Since this is a topic of interest beyond the question of multi-particle localization and we will present a new technique for relating the eigenfunction correlator to the resolvent, we state our results in a fairly general way. 
The basic setup we impose for this part of the analysis deals with the following three basic objects:
\begin{enumerate}
  \item[(H)] The operators $ H(\omega)$ with $ \omega \in \mathbf\Omega $ and $(\mathbf\Omega,\calF,\Prob)$ some probability space form a measurable family of self-adjoint operators in $L^2(\RD)$  (cf.~\cite{CaLa90}*{Sec.~V.1}).
  \item[(O)] $\mathfrak{O}$ is a family of open subsets of $\RD$ and for any $\Omega\in\mathfrak{O}$ there is a self-adjoint measurable restriction $H_\Omega$ of $H$ with domain $\mathcal{D}(H_\Omega)\subset L^2(\Omega)$. Furthermore, the spectrum of $H_\Omega(\omega)$ is only pure point for all $\omega\in\mathbf\Omega$.
  \item[(M)] $\md$ is a pseudo-metric on $\RD$ satisfying
  \begin{equation}
  \sup_{\xbf\in\ZD}\sum_{\ybf\in\ZD}e^{-\mu\md(\xbf,\ybf)^\gamma}<\infty
  \end{equation}
  for all $\mu>0$ and $\gamma\in(0,1]$.
\end{enumerate}  
For our $ n $-particle model presented in the previous section this applies to the following choices:
\begin{description}
\item[\textnormal{ad (H):}]  $D=dn$ and $H=H^{(n)}$.
\item[\textnormal{ad (O):}] 
 The family $\mathfrak{O}$ consists of all sets of the form $\Omega^n$ with an open, bounded set $\Omega\subset\Rd$. The restriction to sets of this form is physically reasonable and furthermore necessary in our analysis, as we make use of the tensor-product structure of $L^2(\Omega^n)=(L^2(\Omega))^n$ and the corresponding identity for the $n$-particle Hamiltonian without interaction.
\item[\textnormal{ad (M):}] The Hausdorff distance $\dist_H$ is the pseudo-metric $\md$ and it is easily checked that the summability condition is fulfilled (cf.~\cite{AW09}*{App. A}). 
The choice of this particular metric is also tied to assumption (GW) below. E.g. for the Euclidean metric the latter does not hold.  
In case $n=1$, of course, $\dist_H$ is just the metric induced induced by the max-norm.
\end{description}

\noindent 
Further assumptions on these basic objects, which we need in order to establish a relation of the eigenfunction correlator to the fractional moment of the resolvent, are: 
  \begin{enumerate}
    \item[(L)] For any $\Omega,\Omega'\in\mathfrak{O}$ with $\Omega'\subset\Omega$ there is a cutoff function $\xi_{\Omega,\Omega'} \in C^\infty(\RD) $, with $0\leq\xi_{\Omega,\Omega'}\leq 1$ and 
  \begin{equation}
  \xi_{\Omega,\Omega'}(\xbf)=
  \begin{cases}
  1 &\text{ if }\xbf\in\Omega'\setminus\Gamma_{\Omega,\Omega'}\\
  0 &\text{ if }\xbf\notin\Omega\setminus\Omega'
  \end{cases},
  \end{equation}
where $\Gamma_{\Omega,\Omega'}=\{\xbf\in\Omega'\,|\,\dist(\xbf,\Omega\cap\partial\Omega')\leq 2\}$ and $ \dist(\cdot,\cdot) $ denotes the maximum distance on $ \RD $, such that for almost all $\omega\in\mathbf\Omega$:
  \begin{enumerate}
  \item[(i)] For any $\phi\in\mathcal{D}(H_\Omega)$:\\ $\xi_{\Omega,\Omega'}\phi\in\mathcal{D}(H_{\Omega'}),\mathcal{D}(H_\Omega)$ and $H_{\Omega'}(\omega)\xi_{\Omega,\Omega'}\phi=H_\Omega(\omega)\xi_{\Omega,\Omega'}\phi$,
  \item[(ii)] $\displaystyle \id_{\Omega\setminus\Gamma_{\Omega,\Omega'}}[H_\Omega(\omega),\xi_{\Omega,\Omega'}]=0$,
  \item[(iii)] For any bounded interval $I$ there is $C_I<\infty$, which is independent of $\Omega$, $\Omega'$ and $ \omega$, such that for any eigenvalue $E\in\sigma(H_\Omega(\omega))\cap I$
  \begin{equation}
  \|[H_\Omega(\omega),\xi_{\Omega,\Omega'}]P_E(H_\Omega(\omega))\| \leq C_I.
  \end{equation}
  \end{enumerate}
  \item[(GW)] For any $\xbf,\ybf\in\RD$ with a minimum distance with respect to $\md$ and for all $\Omega\in\mathfrak{O}$, there is a bounded set $\Omega_{\xbf,\ybf}\in\mathfrak{O}$, $\Omega_{\xbf,\ybf}\subseteq\Omega$, and a sub-$\sigma$-algebra $\calF_{\xbf,\ybf}\subset\calF$ such that for either $(\ubf,\wbf)=(\xbf,\ybf)$ or $(\ubf,\wbf)=(\ybf,\xbf)$:
  \begin{enumerate}
  \item[(i)] $\supp\chi_\ubf\cap\Omega$ is a subset of $\Omega_{\xbf,\ybf}\setminus\Gamma_{\Omega,\Omega_{\xbf,\ybf}}$ and for all
  $ \vbf \in \Gamma_{\Omega,\Omega_{\xbf,\ybf}}$ we have  
  $$
  \md(\ubf,\vbf)\geq \frac{\md(\xbf,\ybf)}{4} \, .
  $$
  Furthermore  the number of unit balls in the $\max$-norm needed to cover $\Gamma_{\Omega,\Omega_{\xbf,\ybf}}$  is polynomially bounded with respect to $\md(\xbf,\ybf)$.
  \item[(ii)] The operator $H_{\Omega_{\xbf,\ybf}}$ is measurable with respect to the $\sigma$-algebra $\calF_{\xbf,\ybf}$.
  \item[(iii)] For any bounded interval $I\subset\R$ there is $C_I>0$ such that
  \begin{equation}\label{eq:GLW}
  \Ex\big[\Tr(\chi_\wbf P_J(H_\Omega))\,\big|\,\calF_{\xbf,\ybf}] \leq C_I|J|,
  \end{equation}
  for any interval $J\subset I$. The constant $C_I>0$ is independent of $\xbf$, $\ybf$ and $\Omega$.
  \end{enumerate}
  \end{enumerate}
  As these assumptions are not that intuitive, let us immediately add a few comments on their relation to our $n$-particle model.
\begin{description}
\item[\textnormal{ad (L):}] This is an assumption of locality of the operator $H$. In the case of the $n$-particle system we choose $\xi_{\Omega,\Omega'}$ such that it satisfies
\begin{equation}\label{eq:cuttoff}
\xi(\xbf) = \begin{cases}
1\qquad\text{if }\xbf\in\Omega'\text{ and }\dist(\xbf,(\partial\Omega')\cap\Omega)\geq 2 \\
0\qquad\text{if } \xbf\in\Omega\text{ and }(\xbf\notin\Omega'\text{ or }\dist(\xbf,(\partial\Omega')\cap\Omega)\leq 1).
\end{cases}
\end{equation}
With this choice, (i) is satisfied for the operator $H=H^{(n)}$. The commutator appearing in (ii) equals zero outside of the support of $\nabla\xi_{\Omega,\Omega'}$. As this support, however, is a subset of $\Gamma_{\Omega,\Omega'}$, the identity (ii) follows. The inequality in (iii) is basically a consequence of the boundedness of the operator $\nabla P_E(H_\Omega)$. In order to make the constant $C_I$ independent of $\Omega$, $\Omega'$ and $E$, the chosen cutoff functions $\xi_{\Omega,\Omega'}$ need to have $L^\infty$-bounds on their first and second derivatives that are independent of $\Omega$ and $\Omega'$. For later purpose, we also note that (L) implies a geometric resolvent equation (cf. \cite{AENSS}*{Lemma 4.2}), i.e.
\begin{equation}
\xi_{\Omega,\Omega'}(H_\Omega-z)^{-1} = (H_{\Omega'}-z)^{-1}\xi_{\Omega,\Omega'}+(H_{\Omega'}-z)^{-1}[H_\Omega,\xi_{\Omega,\Omega'}](H_\Omega-z)^{-1} \label{eq:geomres}
\end{equation}
for any $z\in\C\setminus\R$.

\item[\textnormal{ad (GW)}:] The geometric assumption (i) is intertwined with the probabilistic assumptions (ii) and (iii). For the  $n$-particle model, we choose the appearing terms as follows. If $L=\dist_H(\xbf,\ybf)$, then there is a particle $j\in\{1,\ldots,n\}$ such that $L=\min_k|w_j-u_k|$ with $(\ubf,\wbf)\in\{(\xbf,\ybf),(\ybf,\xbf)\}$. The set $\Omega_{\xbf,\ybf}$ is then chosen as $\big(\bigcup_{k=1}^n \Lambda_{L/2}(u_k)\big)^n$, cf.~Figure~\ref{fig:GLW}. Provided $L$ is sufficiently large, (i) is satisfied. The sub-$\sigma$-algebra $\calF_{\xbf,\ybf}$ is generated by the random variables $\eta_\zeta$, $\zeta\in\Zd\setminus\Lambda_{1+r_U}(w_j)$. Due to this choice, the potential on $\Omega_{\xbf,\ybf}$ is measurable with respect to $\calF_{\xbf,\ybf}$. Lemma~\ref{lemma:exptrproj} then ensures the validity of~(iii).

\end{description}

\begin{figure}[h]
\begin{center}
\includegraphics[width=.7\textwidth]{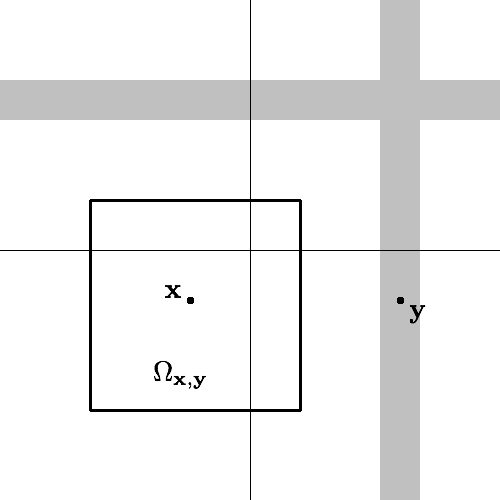}
\end{center}
\caption{Illustration of Assumption~\textnormal{(GW)} in case $ n = 2 $ and $ d= 1 $. For the sketched configurations $ \xbf = (x_1,x_2)$ and $\ybf = (y_1,y_2)$   in the configuration space $ \mathbb{R}^2 $ the Hausdorff distance equals $|x_2-y_1| $.  In this example, it is necessary to choose $ \ubf = \xbf  $ and $ \wbf = \ybf $, since we need to find random variables $ \mathcal{I}' \subset \mathcal{I} $ such that  i)~the 
corresponding potential $N_{ \mathcal{I}'} = \sum_{\zeta \in \mathcal{I}'} N_\zeta $ satisfies the covering condition on 
$B_1(\xbf) $ or  $ B_1(\ybf) $  and ii)~the support of $ N_{ \mathcal{I}'}  $ does not intersect $\Omega_{\xbf,\ybf}  $. 
This is achieved here by random variables $\mathcal{I}_{y_1} $ in the vicinity of $ y_1 $. The shaded area indicates the support of $ N_{\mathcal{I}_{y_1}} $.}\label{fig:GLW}
\end{figure}

\noindent Finally, for one part of our result we only need the following assumptions:
\begin{enumerate}
    \item[(W)] For any bounded interval $I\subset\R$, there exists $C_I>0$ such that
  \begin{equation}
  \sup_{\xbf\in\RD}\sup_{\Omega\in\mathfrak{O}}\Ex\big[\Tr(\chi_\xbf P_J(H_\Omega))] \leq C_I|J|
  \end{equation}
  for any interval $J\subset I$.
  \item[(CT)] For any $\xbf,\ybf\in\RD$, all $\Omega\in\mathfrak{O}$, almost all $\omega\in\mathbf{\Omega}$ and all $z\in\C\setminus\sigma(H_\Omega(\omega))$ 
  \begin{multline}
\|\chi_\xbf(H_\Omega(\omega)-z)^{-1}\chi_\ybf\| \\ \leq \frac{C_0}{\dist(z,\sigma(H_\Omega(\omega)))}\exp\left(\frac{-\mu_0\dist(z,\sigma(H_\Omega(\omega)))}{1+|z|+\dist(z,\sigma(H_\Omega(\omega)))}\md(\xbf,\ybf)\right)
\end{multline}
with uniform constants $C_0,\mu_0>0$.
\end{enumerate}

\noindent
To conclude, let us also  comment on these assumption in relation to our $n $-particle model.
\begin{description}
\item[\textnormal{ad (W)}:] This assumption is a standard Wegner estimate (cf. \cite{Sto01}), which is proven here in \ref{lemma:exptrproj}. We note that it in particular implies the bound
\begin{equation}
  \sup_{\xbf\in\RD}\sup_{\Omega\in\mathfrak{O}}\Ex\big[\Tr(\chi_\xbf f(H_\Omega))] \leq C_I\|f\|_{L^1(I)}
\end{equation}
for any function $f\in L^1(I)$, a fact that we will also use for the similar estimate (GW)(iii).

\item[\textnormal{ad (CT)}] This is a Combes-Thomas estimate in a specific form adjusted to our purposes. As $H^{(n)}(\omega)$ is a Schr\"odinger operator with a bounded potential,  this assumption is  valid thanks  to the results of, e.g., \cites{CT73,GK03}. 

\end{description}

Our first main result in this section is:
\begin{theorem}\label{thm:generalecimpliesfm}
In the setting \textnormal{(H)}, \textnormal{(O)}, \textnormal{(M)}, assume that \textnormal{(W)} and \textnormal{(CT)} hold. Then for any bounded interval $I\subset\R$ and any interval $J\subset I$ with $\dist(J,\partial I)>0$
if there exist $C,\mu\in(0, \infty)$ and $\gamma\in(0,1]$ such that
\begin{equation}
\sup_{\Omega\in\mathfrak{O}}\Ex\bigg[\sum_{E\in\sigma(H_\Omega)\cap I}\|\chi_\xbf P_E(H_\Omega)\chi_\ybf\|\bigg] \leq Ce^{-\mu\md(\xbf,\ybf)^\gamma}
\end{equation}
for all $\xbf,\ybf\in\RD$, then for any $s\in(0,1)$ there exist $A,\nu\in(0, \infty)$ and such that
\begin{equation}
\sup_{\Omega\in\mathfrak{O}}\sup_{\substack{\Real z\in J \\ 0<|\Imag z|< 1}}\Ex\big[\|\chi_\xbf(H_\Omega-z)^{-1}\chi_\ybf\|^s\big] \leq Ae^{-\nu\md(\xbf,\ybf)^\gamma}
\end{equation}
for all $\xbf,\ybf\in\Rd$.
\end{theorem}
\begin{proof}
As the eigenfunction correlator contains only information about $H_\Omega$ restricted to the spectral subspace associated to an interval $I$, it cannot be straightforwardly used to bound a fractional moment of the full resolvent. Accordingly, we need to deal with contributions from energies outside of the interval in a different way. As we will see, it is convenient to introduce a smooth cutoff function $\zeta$ that satisfies $\zeta(E)=0$ for all $E\in J$, where $J\subset I$ is a slightly smaller interval. We then split the resolvent into two parts,
\begin{equation}
\|\chi_\xbf(H_\Omega-z)^{-1}\chi_\ybf\| \leq \bigg\|\chi_\xbf\frac{(1-\zeta)(H_\Omega)}{H_\Omega-z}\chi_\ybf\bigg\| + \bigg\|\chi_\xbf\frac{\zeta(H_\Omega)}{H_\Omega-z}\chi_\ybf\bigg\|, \label{resspecdecomp}
\end{equation}
the first of which is dealt with by using the eigenfunction correlator bound and the second of which is bounded by means of the following lemma. Its proof requires assumption (CT) and is deferred to \ref{appendix:hsct}.

\begin{lemma}
Assume that \textnormal{(H)}, \textnormal{(O)} and \textnormal{(CT)} hold and let $J\subset\R$ be a bounded interval, $\Omega\in\mathfrak{O}$, $\omega\in\mathbf{\Omega}$ and $\delta>0$. For any $\xbf,\ybf\in\RD$ there is a function $\zeta= \zeta_{\xbf,\ybf} \in C^\infty(\R)$ with $0\leq\zeta\leq 1$ such that $\zeta(x)=0$ if $\dist(x,J)\leq\delta$, $\zeta(x)=1$ if $\dist(x,J)\geq 2\delta$ and that for all $z\in\C\setminus\sigma(H_\Omega(\omega))$
\begin{equation}
\bigg\|\chi_\xbf\frac{\zeta(H_\Omega(\omega))}{H_\Omega(\omega)-z}\chi_\ybf\bigg\| \leq Ce^{-\mu\md(\xbf,\ybf)} \label{ctandhsbound} 
\end{equation}
holds with constants $C,\mu\in(0,\infty)$ that can be chosen independently of $\xbf$, $\ybf$, $z$, $\Omega$ and $\omega$.
\end{lemma}

In our setting, we choose $\delta>0$ such that the eigenfunction correlator bound holds on the larger interval $I=[\inf J-2\delta,\inf J+2\delta]$. Thus, as long as $\Real z\in J$, the second summand in \eqref{resspecdecomp} can be bounded deterministically (and hence in a fractional moment) by \eqref{ctandhsbound}. We note that the cutoff function depends on $\xbf$ and $\ybf$, the reason of which can be seen in the proof of the lemma. However, in our setting, this peculiarity does not pose a problem.

It remains to find a bound on the first term in \eqref{resspecdecomp},
\begin{equation}
\bigg\|\chi_\xbf\frac{(1-\zeta)(H_\Omega)}{H_\Omega-z}\chi_\ybf\bigg\| \leq \sum_{E\in\sigma(H_\Omega)\cap I} \frac{\|\chi_\xbf P_E(H_\Omega)\chi_\ybf\|}{|E-z|} \, .
\end{equation}
As we will see, this term is easier to handle if $\|\chi_\xbf P_E(H_\Omega)\chi_\ybf\|$ is replaced by its square. This modification can be justified as follows.
The inequalities
\begin{align}
&\Ex\bigg[\bigg(\sum_{E\in\sigma(H_\Omega)\cap I} \frac{\|\chi_\xbf P_E(H_\Omega)\chi_\ybf\|}{|E-z|}\bigg)^s\bigg] \notag \\
  &\leq \sum_{\wbf\in\ZD}\Ex\bigg[\bigg(\sum_{E\in\sigma(H_\Omega)\cap I} \mkern-12mu\frac{\|\chi_\xbf P_E(H_\Omega)\chi_\wbf\|^2}{|E-z|}\bigg)^s\bigg]^{\frac{1}{2}}  \Ex\bigg[\bigg(\sum_{E\in\sigma(H_\Omega)\cap I} \mkern-12mu\frac{\|\chi_\wbf P_E(H_\Omega)\chi_\ybf\|^2}{|E-z|}\bigg)^s\bigg]^{\frac{1}{2}}
\end{align}
together with assumption (M) imply that the left side 
decays (sub)exponentially provided both two terms on the right side do so for some $s\in(0,1)$. In order to prove a bound on the terms on the right side, we choose $q>1$ such that $1-s<sq<1$ and let $p>1$ be the conjugate H\"older exponent, i.e. $p^{-1}+q^{-1}=1$. Then a two-fold application of H\"older's inequality yields
\begin{align}
&\Ex\bigg[\bigg(\sum_{E\in\sigma(H_\Omega)\cap I}\frac{\|\chi_\xbf P_E(H_\Omega)\chi_\ybf\|^2}{|E-z|}\bigg)^s\bigg] \notag\\
&\leq \Ex\bigg[\bigg(\sum_{E\in\sigma(H_\Omega)\cap I}   \mkern-12mu\|\chi_\xbf P_E(H_\Omega)\chi_\ybf\|^{\big(2-\frac{1}{sq}\big)p}\bigg)^s\bigg]^{\frac{1}{p}}\Ex\bigg[\bigg(\sum_{E\in\sigma(H_\Omega)\cap I} \mkern-12mu \frac{\|\chi_\xbf P_E(H_\Omega)\chi_\ybf\|^{\frac{q}{sq}}}{|E-z|^q}\bigg)^s\bigg]^{\frac{1}{q}} \notag \\
&\leq \Ex\bigg[\sum_{E\in\sigma(H_\Omega)\cap I}\|\chi_\xbf P_E(H_\Omega)\chi_\ybf\|\bigg]^{\frac{s}{p}}\Ex\bigg[\sum_{E\in\sigma(H_\Omega)\cap I}\frac{\|\chi_\xbf P_E(H_\Omega)\chi_\ybf\|}{|E-z|^{sq}}\bigg]^{\frac{1}{q}},
\end{align}
where we used that $(2-1/sq)p\geq 1$ and $\|\chi_\xbf P_E(H_\Omega)\chi_\ybf\|\leq 1$. The first factor contains the eigenfunction correlator and decays (sub)exponentially in the distance of $\xbf$ and $\ybf$. As for the second term, we first estimate $ \|\chi_\xbf P_E(H_\Omega)\chi_\ybf\| \leq \sqrt{\|\chi_\xbf P_E(H_\Omega)\chi_\xbf\| \|\chi_\ybf P_E(H_\Omega)\chi_\ybf\| } $ and then use the Cauchy-Schwarz inequality with respect to the summation and expectation. It hence remains to 
 estimate 
\begin{align}
& \Ex\bigg[\sum_{E\in\sigma(H_\Omega)\cap I}\frac{\|\chi_\xbf P_E(H_\Omega)\chi_\xbf\|}{|E-z|^{sq}}\bigg] \leq \Ex\bigg[\sum_{E\in\sigma(H_\Omega)\cap I}\frac{\Tr(\chi_\xbf P_E(H_\Omega)\chi_\xbf)}{|E-z|^{sq}}\bigg] \notag \\
&\qquad \leq 2\liminf_{\eta\downarrow0}\int_I |E-z|^{-sq}\Ex\big[\Tr(\chi_\xbf f_\eta(H_\Omega-E)\chi_\xbf)\big]\dE \, , 
\end{align}
and likewise for $ \xbf \leftrightarrow \ybf $. In the last step $f_\eta(\lambda)=f(\lambda/\eta)/\eta$ denotes 
a suitable approximating $\delta$-function with compact support. An upper bound now follows from Assumption (W) (using that, for small $\eta>0$, $\|f_\eta\|_{L^1(I')}=1$ with a slightly enlarged interval $I'\supset I$) and the fact that $sq<1$.
\end{proof}

In order to establish eigenfunction correlator decay using that of the fractional moment  we have
\begin{theorem}\label{thm:generalfmimpliesec}
In the setting \textnormal{(H)}, \textnormal{(O)}, \textnormal{(M)}, assume that \textnormal{(L)} and \textnormal{(GW)} hold. Then for any bounded interval $I\subset\R$ if there exist $A,\mu\in(0,\infty)$, $s\in(0,1)$ and $\gamma\in(0,1]$ such that
\begin{equation}\label{eq:fracmomlocgen}
\sup_{\Omega\in\mathfrak{O}}\sup_{\substack{\Real z\in I \\ 0<|\Imag z|< 1}}\Ex\big[\|\chi_\xbf(H_\Omega-z)^{-1}\chi_\ybf\|^s\big] \leq Ae^{-\nu\md(\xbf,\ybf)^\gamma}
\end{equation}
for all $\xbf,\ybf\in\RD$, then there exist $C,\mu\in(0,\infty)$ such that
\begin{equation}
\sup_{\Omega\in\mathfrak{O}}\Ex\bigg[\sum_{E\in\sigma(H_\Omega)\cap I}\|\chi_\xbf P_E(H_\Omega)\|\|\chi_\ybf P_E(H_\Omega)\|\bigg] \leq Ce^{-\mu\md(\xbf,\ybf)^\gamma}
\end{equation}
for all $\xbf,\ybf\in\RD$.
\end{theorem}
This theorem is also of interest in the one-particle case for which it constitutes an alternative to the steps taken in~\cite{AENSS}. In this case, the locality requirement~\textnormal{(L)} is clearly satisfied with $ \mathfrak{O} $ the set of all bounded, open sets in $ \mathbb{R}^d $. 
Assumption~\textnormal{(GW)}  with $ d $ the Euclidean distance boils down to the "independence at a distance"  (cf.~\cite{AENSS}) of the  basic random variables together with a Wegner estimate involving local averages only.

\begin{proof}[Proof of Theorem~\ref{thm:generalfmimpliesec}]
Similarly as in the proof of Theorem~\ref{thm:generalecimpliesfm}, it is more convenient prove a bound on
\begin{equation}
\Ex\bigg[\sum_{E\in\sigma(H)\cap I}\|\chi_\xbf P_E(H_\Omega)\|^2\|\chi_\ybf P_E(H_\Omega)\|^2\bigg]. \label{efcorrsq}
\end{equation}
The additional exponent of $2$ can be justified as before. 

We apply Assumption (L) and (GW) with $\Omega'=\Omega_{\xbf,\ybf}$ and infer for any eigenvalue $E\in I$ of $H_\Omega$ and a fixed $\eps\in(0,1)$
\begin{align}
\chi_{\ubf}P_E(H_\Omega) &= \chi_\ubf\xi_{\Omega,\Omega_{\xbf,\ybf}}P_E(H_\Omega) \notag \\
& = \chi_\ubf (H_{\Omega_{\xbf,\ybf}}-E-i\eps)^{-1}(H_{\Omega_{\xbf,\ybf}}-E-i\eps)\xi_{\Omega,\Omega_{\xbf,\ybf}}P_E(H_\Omega) \notag\\
&=\chi_\ubf (H_{\Omega_{\xbf,\ybf}}-E-i\eps)^{-1}([H_\Omega,\xi_{\Omega,\Omega_{\xbf,\ybf}}] - i \eps )\, \xi_{\Omega,\Omega_{\xbf,\ybf}}P_E(H_\Omega)
\end{align}
and hence
\begin{equation}
\|\chi_\ubf P_E(H_\Omega)\| \leq C_I\|\chi_\ubf (H_{\Omega_{\xbf,\ybf}}-E-i\eps)^{-1}\id_{\Gamma_{\Omega,\Omega_{\xbf,\ybf}}}\| + \eps\|\chi_\ubf (H_{\Omega_{\xbf,\ybf}}-E-i\eps)^{-1}\|.
\end{equation}
We insert this into \eqref{efcorrsq} and use that $\|\chi_\ubf P_E(H_\Omega)\|^2\leq\|\chi_\ubf P_E(H_\Omega)\|^s$, such that 
\begin{multline}
\Ex\bigg[\sum_{E\in\sigma(H_\Omega)\cap I}\|\chi_\xbf P_E(H_\Omega)\|^2\|\chi_\ybf P_E(H_\Omega)\|^2\bigg] \\
\leq C_I^s \, \Ex\bigg[\sum_{E\in\sigma(H_\Omega)\cap I}  \|\chi_\ubf(H_{\Omega_{\xbf,\ybf}}-E-i\eps)^{-1}\id_{\Gamma_{\Omega,\Omega_{\xbf,\ybf}}}\|^s\|\chi_\wbf P_E(H_\Omega)\|^2\bigg] \\
+\eps^s\, \Ex\bigg[\sum_{E\in\sigma(H_\Omega)\cap I}\|\chi_\ubf(H_{\Omega_{\xbf,\ybf}}-E-i\eps)^{-1}\|^s\|\chi_\wbf P_E(H_\Omega)\|^2\bigg]. \label{efdecayinefcorrsq}
\end{multline}
Now let $f_\eta(\lambda)=f(\lambda/\eta)/\eta$ be a compactly supported approximating $\delta$-function. Then
\begin{multline}
\Ex\bigg[\sum_{E\in\sigma(H_\Omega)\cap I}  \|\chi_\ubf(H_{\Omega_{\xbf,\ybf}}-E-i\eps)^{-1}\id_{\Gamma_{\Omega,\Omega_{\xbf,\ybf}}}\|^s\|\chi_\wbf P_E(H_\Omega)\|^2\bigg] \\
\leq 2\liminf_{\eta\downarrow0}\int_I \Ex\big[\|\chi_\ubf(H_{\Omega_{\xbf,\ybf}}-E-i\eps)^{-1}\id_{\Gamma_{\Omega,\Omega_{\xbf,\ybf}}}\|^s\Tr(\chi_\wbf f_\eta(H_\Omega-E))\big]\dE  \, . 
\end{multline}
We use that $H_{\Omega_{\xbf,\ybf}}$ is measurable with respect to $\calF_{\xbf,\ybf}$ due to Assumption (GW)(ii) and conclude
\begin{multline}
\Ex\big[\|\chi_\ubf(H_{\Omega_{\xbf,\ybf}}-E-i\eps)^{-1}\id_{\Gamma_{\Omega,\Omega_{\xbf,\ybf}}}\|^s\Tr(\chi_\wbf f_\eta(H_\Omega-E))\big] \\
= \Ex\big[\|\chi_\ubf(H_{\Omega_{\xbf,\ybf}}-E-i\eps)^{-1}\id_{\Gamma_{\Omega,\Omega_{\xbf,\ybf}}}\|^s\Ex\big[\Tr(\chi_\wbf f_\eta(H_\Omega-E))\big|\calF_{\xbf,\ybf}\big]\big] \\
\leq C_I\, \Ex\big[\|\chi_\ubf(H_{\Omega_{\xbf,\ybf}}-E-i\eps)^{-1}\id_{\Gamma_{\Omega,\Omega_{\xbf,\ybf}}}\|^s\big].
\end{multline}
We finally use the decay estimate on fractional moments and infer
\begin{equation}
\Ex\big[\|\chi_\ubf(H_{\Omega_{\xbf,\ybf}}-E-i\eps)^{-1}\id_{\Gamma_{\Omega,\Omega_{\xbf,\ybf}}}\|^s\big] \leq Ce^{-\mu\md(\xbf,\ybf)^\gamma}.
\end{equation}
This is possible since we assumed in (GW)(i) that the set $\Gamma_{\Omega,\Omega_{\xbf,\ybf}}$ can be covered by unit balls around points $\vbf$ satifying $\md(\vbf,\ubf)\geq\md(\xbf,\ybf)/4$ and the number of which is bounded polynomially in $\md(\xbf,\ybf)$. This yields a (sub)exponential bound on the first term in \eqref{efdecayinefcorrsq}. The second term vanishes in the limit $\eps\downarrow0$.
\end{proof}

The preceding theorems do not contain any localization results for the infinite volume operator $H$ (unless $\RD\in\mathfrak{O}$). In order to conclude such results from fractional moment bounds as in \eqref{eq:fracmomlocgen}, we need an additional assumption:
\begin{enumerate}
\item[(E)] There is a sequence of sets $\Omega_L\in\mathfrak{O}$, $ L \in \mathbb{N} $, which exhaust $\RD$ in the sense that for all $\xbf\in\RD$ there exists $L_\xbf$ such that $B_{1/2}(\xbf)\subset\Omega_L$ for all $L\geq L_\xbf$, and almost surely:
\begin{enumerate}
\item[(i)] $H_{\Omega_L}$ converges to $H$ in strong resolvent sense.
\item[(ii)] $\id_{\Omega_L}P_J(H)$ is compact for all $L$ and all bounded intervals $J\subset\R$.
\end{enumerate}
\end{enumerate}
This assumption is evidently satisfied in the case of the operator $H^{(n)}$, for which on may choose $\Omega_L:=(\Lambda_L(0))^n$.

\begin{theorem}\label{thm:locinfvolume}
In the setting \textnormal{(H)}, \textnormal{(O)} and \textnormal{(M)}, assume \textnormal{(E)} and that \textnormal{(L)} and \textnormal{(GW)} also hold for $\Omega=\RD$. Suppose $I\subset\R$ is a bounded interval and that there exist $A,\mu\in(0,\infty)$, $s\in(0,1)$ and $\gamma\in(0,1]$ such that
\begin{equation}
\sup_{\Omega\in\mathfrak{O}}\sup_{\substack{\Real z\in I \\ 0<|\Imag z|< 1}}\Ex\big[\|\chi_\xbf(H_\Omega-z)^{-1}\chi_\ybf\|^s\big] \leq Ae^{-\nu\md(\xbf,\ybf)^\gamma}
\end{equation}
for all $\xbf,\ybf\in\RD$. Then the following holds:
\begin{enumerate}
\item[\textnormal{(i)}] There exist $C,\mu\in(0,\infty)$ such that for all $\xbf,\ybf\in\RD$
\begin{equation}
\Ex\big[\sup_{|f|\leq 1}\|\chi_\xbf f(H)P_I(H)\chi_\ybf\|\big] \leq Ce^{-\mu\md(\xbf,\ybf)^\gamma},
\end{equation}
where the supremum is taken over all measurable functions on $I$ that are in modulus uniformly bounded by one.
\item[\textnormal{(ii)}] The operator $H$ has almost surely only pure point spectrum in $I$ and
\begin{equation}
\Ex\bigg[\sum_{E\in\sigma(H)\cap I}\|\chi_\xbf P_E(H)\|\|\chi_\ybf P_E(H_\Omega)\|\bigg] \leq Ce^{-\mu\md(\xbf,\ybf)^\gamma} \label{eq:efcorrRD}
\end{equation}
for some $C,\mu\in(0,\infty)$ and all $\xbf,\ybf\in\RD$.
\item[\textnormal{(iii)}] Let $g:\ZD\rightarrow(1,\infty)$ be a function with $\sum_{\xbf\in\ZD}g(\xbf)^{-1}=1$. Then all normalized eigenfunctions $\phi$ of $H$ with eigenvalue $E\in I$ satisfy
\begin{equation}
\|\chi_\xbf\phi\| \leq A(\omega)\frac{g(\xbf_\phi)}{\sqrt{\alpha_\phi}}e^{-\mu'\md(\xbf_\phi,\xbf)^\gamma} \label{eq:efbound}
\end{equation}
for all $ \xbf \in \ZD $, all $\mu'\in(0,\mu)$ and some $A\in L^1(\mathbf{\Omega})$. Here $\alpha_\phi:=\sum_{\ubf\in\ZD}g(\ubf)^{-1}\|\chi_\ubf\phi\|$ and $\xbf_\phi\in\ZD$ is chosen such that $\|\chi_{\xbf_\phi}\|=\max_{\ubf\in\ZD}\|\chi_\ubf\phi\|$.
\end{enumerate}
\end{theorem}
\begin{proof}
\begin{enumerate}
\item[(i)] It follows directly from the estimate on the eigenfunction correlator in Theorem~\ref{thm:generalfmimpliesec} that for all sets $\Omega\in\mathfrak{O}$
\begin{equation}
\Ex\bigg[\sup_{|f|\leq 1}\|\chi_\xbf f(H_\Omega)P_I(H_\Omega)\chi_\ybf\|\bigg] \leq Ce^{-\mu\md(\xbf,\ybf)^\gamma}\, .
\end{equation}
This estimate can be extended to the whole space $\Omega=\RD$. In the case of continuous functions $f$ with compact support, one can use that $f(H_{\Omega_L})$ converges in strong resovent sense to $H$. An approximation argument extends the result to general measurable functions, cf. \cite{AENSS}*{Section 2.5}.
\item[(ii)] From (M), (E) and (i) it follows (with $f(E)=e^{-itE}$) that for all $\xbf\in\RD$
\begin{equation}
\lim_{L\rightarrow\infty}\Ex\big[\|\id_{\Omega_L^c}e^{-itH}P_I(H)\chi_\xbf\|^2\big] \leq \lim_{L\rightarrow\infty} \mkern-10mu \sum_{\substack{\ubf\in\ZD\\ B_{\frac{1}{2}}(\ubf) \cap \Omega_L^c \neq \emptyset}} \mkern-10mu \Ex\big[\|\chi_\ubf e^{-itH}P_I(H)\chi_\xbf\|\big] = 0.
\end{equation}
Using the RAGE theorem (cf. \cite{Sto01}*{Theorem 4.1.20}), we conclude that the spectrum of $H$ in $I$ is almost surely only pure point. Now that we have established pure point spectrum also for the infinite volume operator $H$, we conclude inequality \eqref{eq:efcorrRD} using assumptions (L) and (GW) for $\Omega=\RD$ along the same lines of reasoning as the proof of Theorem \ref{thm:generalfmimpliesec}.
\item[(iii)] Defining
\begin{equation}
A:=\sum_{\xbf\in\ZD}g(\xbf)^{-1}\sum_{\ybf\in\ZD}\|\chi_\xbf P_E(H)\|\|\chi_\ybf P_E(H)\|e^{-\mu'\md(\xbf,\ybf)^\gamma},
\end{equation}
we have $\Ex[A]<\infty$ by (ii) and (M). Therefore, for any normalized eigenfunction $\phi$ with eigenvalue $E\in I$
\begin{equation*}
\|\chi_\xbf\phi\|\|\chi_\ybf\phi\| \leq Ag(\xbf)e^{-\mu'\md(\xbf,\ybf)^\gamma}
\end{equation*}
for all $\xbf,\ybf\in\ZD$. Inequality \eqref{eq:efbound} follows by an adaption of \cite{GT12}*{Lemma 3.5} to our setting.
\end{enumerate}
\end{proof}

Going back to our $n$-particle random Schr\"odinger, Theorems \ref{thm:generalecimpliesfm} and \ref{thm:generalfmimpliesec} imply the following:

\begin{corollary}\label{thm:equivalenceoffmandec}
Let $I\subset\R$ be a bounded interval.
\begin{enumerate}
\item[\textnormal{(i)}] If $I$ is a regime of $n$-particle fractional moment localization of order $\gamma\in(0,1]$, then it is also a regime of $n$-particle eigenfunction correlator localization of order $\gamma\in(0,1]$.
\item[\textnormal{(ii)}] If $I$ is a regime of $n$-particle eigenfunction correlator localization of order $\gamma\in(0,1]$, then any interval $J\subset I$ satisfying $\dist(J,\partial I)>0$ is a regime of $n$-particle fractional moment localization of order $\gamma\in(0,1]$.
\end{enumerate}
\end{corollary}

\section{Proof of multiparticle localization}\label{section:mploc}

\subsection{Localization in systems consisting of non-interacting subclusters}\label{subsection:locfornonintsubclusters}

As a first step, we consider $n$-particle systems which can be divided into two non-interacting subclusters. Suppose we have a partition $\{1,\ldots,n\}=J\dot\cup K$ with non-empty sets $J$ and $K$. If there is no interaction between the particles with numbers in $J$ and particles with numbers in $K$, respectively, the corresponding Hamiltonian takes the form
\begin{equation}\label{eq:HJK}
H^{(J,K)} = H^{(\# J)}\otimes\id + \id\otimes H^{(\# K)}
\end{equation}
where the tensor product is to be understood in the sense that $H^{(\# J)}$ acts on the variables $x_j$, $j\in J$, and $H^{(\# K)}$ acts on the variables $x_k$, $k\in K$. Now assume that for any $m\leq n-1$, $m$-particle fractional moment localization has been proven in the interval $I^{(m)}=[E_0^{(m)},E_0^{(m)}+\eta^{(n-1)}]$. Then it is natural to assume that we can also prove localization of $H^{(J,K)}$ for energies near $E_0^{(n)}$, as both its subsystems are localized at the bottom of the spectrum. However, we cannot directly conclude fractional moment localization for $H^{(J,K)}$, but instead take a detour via eigenfunction correlator localization. The drawback of this method is that the length of the interval of localization becomes slightly smaller. The decay of fractional moments can be proven not only with respect to the ``usual'' $n$-particle Hausdorff distance $\dist_H$ but with respect to a distance $\dist_H^{(J,K)}(\geq \dist_H)$ that takes into account the decomposition into subclusters. It is defined by
\begin{equation}\dist_H^{(J,K)}(\xbf,\ybf)=\max\{\dist_H(\xbf_J,\ybf_J),\dist_H(\xbf_K,\ybf_K)\},
\end{equation}
where $\xbf_J=(x_j)_{j\in J}$, $\xbf_K=(x_k)_{k\in K}$ etc. The decay of fractional moments with respect to this distance will be necessary for the analysis in the following section.

\begin{theorem}\label{thm:fmboundnonint}
Assume that there exists $\eta^{(n-1)}>0$ such that for any $m\in\{1,\ldots,n-1\}$ the operator $H^{(m)}$ exhibits fractional moment localization of order $\zeta\in(0,1]$ in the interval $I^{(m)}=[E^{(m)},E_0^{(m)}+\eta^{(n-1)}]$. Then for any $\eta^{(n)}\in[0,\eta^{(n-1)})$, any $s\in(0,1)$ and any partition $J\dot\cup K=\{1,\ldots,n\}$
\begin{equation}
\sup_{\substack{\Omega\subset\Rd \\ \mathrm{open,\ bd.} }}\sup_{\substack{\Real z\in I^{(n)} \\ 0<|\Imag z|<1}}\Ex\big[\|\chi_\xbf(H_{\Omega^n}^{(J,K)}-z)^{-1}\chi_\ybf\|^s\big] \leq Ce^{-\mu\dist_H^{(J,K)}(\xbf,\ybf)^\zeta} \label{eq:thm:fmboundnonint}
\end{equation}
for some $C,\mu>0$ and all $\xbf,\ybf\in\Rdn$, where $I^{(n)}=[E_0^{(n)},E_0^{(n)}+\eta^{(n)}]$.
\end{theorem}

\begin{proof}
Due to Theorem \ref{thm:generalfmimpliesec} (more specifically Corollary \ref{thm:equivalenceoffmandec}), we have 
\begin{equation}
\Ex\bigg[\sum_{E\in\sigma(H_{\Omega^n}^{(\# J)})\cap I^{(\# J)}}\|\chi_{\xbf_J}P_E(H_{\Omega^{(\# J)}}^{(\# J)})\chi_{\ybf_J}\|\bigg] \leq C_Je^{-\nu_J\dist_H(\xbf_J,\ybf_J)^\zeta}  \label{nointEC1}
\end{equation}
for all bounded open sets $\Omega\subset\Rd$ and all $\xbf_J,\ybf_J\in\R^{d(\#J)}$. An analogous inequality holds for $H_{\Omega^{(\# K)}}^{(\# K)}$. For the tensor product sum $H^{(J,K)}$ of these operators, we have
\begin{multline}
\sum_{E\in\sigma(H_{\Omega^n}^{(J,K)})\cap \tilde{I}^{(n)}}\|\chi_{\xbf}P_E(H_{\Omega^n}^{(J,K)})\chi_{\ybf}\| \\
\leq \sum_{\substack{E_J\in\sigma(H_{\Omega^n}^{(\# J)}) \\ E_K\in\sigma(H_{\Omega^n}^{(\# K)}) \\ E_J+E_K\in \tilde{I}^{(n)}}}\|(\chi_{\xbf_J}\otimes\chi_{\xbf_K})(P_{E_J}(H_{\Omega^{\# J}}^{(\#J)})\otimes P_{E_K}(H_{\Omega^{\# K}}^{(\#K)}))(\chi_{\ybf_J}\otimes\chi_{\ybf_K})\|,
\end{multline}
where $\tilde{I}^{(n)}=[E_0^{(n)},E_0^{(n)}+\eta^{(n-1)}]$. As $E_0^{(n)}\leq E_0^{(\# J)}+E_0^{(\# K)}$ by Proposition~\ref{prop:gs}, all eigenvalues $E_J\in\sigma(H_{\Omega^{(\# J)}}^{(\# J)})$ and $E_K\in\sigma(H_{\Omega^{(\# K)}}^{(\# K)})$ satisfying $E_J+E_K\in \tilde{I}^{(n)}$  belong to the intervals $I^{(\# J)}$ and $I^{(\# K)}$, respectively. The estimate \eqref{nointEC1} and its analogue for $H^{(\# K)}$ thus allow us to infer
\begin{align}
&\Ex\bigg[\sum_{E\in\sigma(H_{\Omega^n}^{(J,K)})\cap \tilde{I}^{(n)}}\|\chi_\xbf P_E(H_{\Omega^n}^{(J,K)})\chi_\ybf\|\bigg] \notag\\
&\leq \Ex\bigg[\sum_{E_J\in\sigma(H_{\Omega^{\# J}}^{(\# J)})\cap I^{(\# J)}}\mkern-25mu\|\chi_{\xbf_J} P_{E_J}(H_{\Omega^{\# J}}^{(\# J)})\chi_{\ybf_J}\| 
\mkern-10mu\sum_{E_K\in\sigma(H_{\Omega^{\# K}}^{(\# K)})\cap I^{(\# K)}}\mkern-25mu\|\chi_{\xbf_K} P_{E_K}(H_{\Omega^n}^{(\# K)})\chi_{\ybf_K}\|\bigg] \notag \\
&\leq C \, e^{-\nu\max\{\dist_H(\xbf_J,\ybf_J),\dist_H(\xbf_K,\ybf_K)\}^\zeta}.
\end{align}
In the last step we used that each of the sums can be bounded deterministically, i.e., for all $ m $ with suitably large $ a, p > 0 $:
\begin{align}
\sum_{E\in\sigma(H_{\Omega^{m}}^{(m)})\cap I} \left\| \chi_\xbf P_{E}\big(H_{\Omega^{m}}^{(m)}(\omega)\big)  \chi_\xbf  \right\| & \leq \Tr \chi_\xbf P_{I}\big(H_{\Omega^{m}}^{(m)}(\omega)\big) \notag \\
&  \leq ( a + \sup I )^{2p} \sup_{\omega} \Tr \chi_\xbf  \big(H_{\Omega^{m}}^{(m)}(\omega)+ a \big)^{-2p} \, . 
\end{align}
Theorem \ref{thm:generalecimpliesfm} with $ \dist_H^{(J,K)} $ as the pseudo-metric then shows that \eqref{eq:thm:fmboundnonint} holds. In this case, we do not have to cut off a part of $\tilde{I}^{(n)}$ at the lower endpoint as it is located at the bottom of the spectrum.
\end{proof}

\subsection{A rescaling inequality}

Building on the result of the previous section, we will now deal with the fully-interacting $n$-particle system. For this purpose, we define
\begin{equation*}
B_s^{(n)}(I,L):=\sup_{\substack{\xbf,\ybf\in\Rdn \\ \dist_H(\xbf,\ybf)\geq L}}\sup_{\substack{\Omega\subset\Rd \\ \mathrm{open,\, bd.}}}\sup_{\substack{\Real z\in I \\ 0<|\Imag z|<1}}\Ex\big[\|\chi_\xbf(H_{\Omega^n}^{(n)}-z)^{-1}\chi_\ybf\|^s\big]
\end{equation*}
for any $s\in(0,1)$, $L\geq0$ and any bounded interval $I\subset\R$. Clearly, (sub)exponential decay of this quantity with respect to $L$ is just another way to express fractional moment localization in $I$ (with the same rate of decay). The conclusion of Theorem~\ref{thm:fmboundnonint} can then be formulated as
\begin{equation}
B_s^{(J,K)}(I,L) \leq A^*e^{-\nu^*L^{\gamma^*}}, \label{assumption:locfornonintsubclusters}
\end{equation}
where $B_s^{(J,K)}$ is defined in the same way as $B_s^{(n)}$, but with $H^{(n)}$ and $\dist_H$ replaced by $H^{(J,K)}$ and $\dist_H^{(J,K)}$, respectively. For the analysis in this section, it is not of importance if the interval $I$ is located at the bottom of the spectrum of the Hamiltonian. We will thus simply assume that \eqref{assumption:locfornonintsubclusters} holds for some bounded interval $I$, some constants $A^*,\nu^*>0$, $\gamma^*\in(0,1]$ and all partitions $J\dot\cup K=\{1,\ldots,n\}$.

The key theorem is the following rescaling inequality. For technical purposes, we introduce a "safety distance'' $R=r_U+6$ which will be needed in some of our arguments.

\begin{theorem}\label{thm:rescalingineq}
Let $I$ be a bounded interval such that \eqref{assumption:locfornonintsubclusters} holds for some $s\in(0,1/3)$, some $A^*,\nu^*>0$, $\gamma^*\in(0,1]$, all $ L \geq 0 $ and all partitions $J\dot\cup K=\{1,\ldots,n\}$. Then there are constants $C,\nu_2>0$ such that for any $\alpha\in(0,1]$ and all sufficiently large $L$
\begin{multline}
B_s^{(n)}(I,2L+2L^\alpha+9R) \leq C\bigg(L^{8dn}B_s^{(n)}(I,L+L^\alpha)^2+e^{-\nu_2 L^{\alpha\gamma^*}} \\ +L^{(5+\alpha)nd}\bigg(w_b\bigg(\frac{L^\alpha}{4n}\bigg)\bigg)^s B_s^{(n)}(I,2L)\bigg). \label{eq:rescalingineq}
\end{multline}
The constant $C>0$ is independent of $L$.
\end{theorem}

In the case of $n=1$, the assumption that \eqref{assumption:locfornonintsubclusters} holds is meaningless. The theorem is then understood to hold for any bounded interval $I$. As will become clear in the proof, the second and the third summand in the bound do not appear in the case $n=1$.

Once the theorem is proved, one just needs show that $B_s^{(n)}(I,L)$ is sufficiently small on suitable initial lengths $L\in [L_1,4L_1]$. The following corollary then yields (sub)exponential decay of $B_s^{(n)}$. The initial estimates for different regimes of energies and interaction are the subject of Section \ref{section:initialestimates}.

\begin{corollary}\label{cor:rescalingimpliesdecay}
In the situation of Theorem \ref{thm:rescalingineq}, the following holds:
\begin{enumerate}
  \item[(i)] Assume that $w_b(r)\leq c_w r^{-p_w}$, where $p_w$ satisfies $\alpha p_w s > (5+\alpha)nd$ with $\alpha\in(0,1)$. Furthermore, let $C'>0$, $q'>8dn$. Then for all sufficiently large $L_1$ the following holds: If $B_s^{(n)}(I,L)\leq C'L^{-q'}$ for all $L\in[L_1,4L_1]$, then
  \begin{equation}
  B_s^{(n)}(I,L)\leq 2C'e^{-\nu' L^\beta}
  \end{equation}
  for some $\nu'>0$ and all $L\geq L_1$ with $\beta=\min\{\alpha\gamma^*,1-\alpha\}$.
\item[(ii)] Assume that $w_b(r)\leq c_w e^{-\mu_w r^{\gamma^*}}$ and let $C'>0$, $q'>8dn$. Then for all sufficiently large $L_1$ the following holds: If $B_s^{(n)}(I,L)\leq C'L^{-q'}$ for all $L\in[L_1,4L_1+9R]$, then
  \begin{equation}
  B_s^{(n)}(I,L)\leq 2C'e^{-\nu' L^{\gamma*}}
  \end{equation}
  for some $\nu'>0$ and all $L\geq L_1$.
\end{enumerate}
\end{corollary}

\begin{proof}
\begin{itemize}
  \item[(i)] We abbreviate $B_s^{(n)}(I,L)$ by $B_L$.\\
  Let $\beta=\min\{\alpha\gamma^*,1-\alpha\}$ and $\nu' = \min\{ \ln 2 / (4L_1)^\beta , \nu_2/(5^\beta 2))$. Our assumption implies that   for all $L\in[L_1,4L_1]$:
  \begin{equation}
  B_L \leq 2C'e^{-\nu'L^\beta}L^{-q'} \, . \label{pf:cor:rescalingimpliesdecay:1}
  \end{equation}
Assuming that $ L_1 $ is sufficiently large we will prove that \eqref{eq:rescalingineq} implies that \eqref{pf:cor:rescalingimpliesdecay:1} holds for all $L$ in the larger interval $[L_1,4L_1+2\cdot2^\alpha L_1^\alpha+9R]$. Iteration of this step then yields the result for all $L\geq L_1$.

  Assume that $L\in[L_1,2L_1]$. Then \eqref{eq:rescalingineq} and \eqref{pf:cor:rescalingimpliesdecay:1} imply
  \begin{align}
  &\frac{(2L+2L^\alpha+9R)^{q'}e^{\nu'(2L+2L^\alpha+9R)^\beta}}{2C'} B_{2L+2L^\alpha+9R} \notag \\
  &\leq 5^{q'}2C'C L^{8dn-q'}e^{\nu'(2L+2L^\alpha+9R)^\beta-2\nu'(L+L^\alpha)^\beta} \notag\\
  &\quad + \frac{C}{2C'}5^{q'}L^{q'}e^{\nu'5^\beta L^\beta-\nu_2 L^{\alpha\gamma^*}} \notag\\
  &\quad +3^{q'}C L^{(5+\alpha)nd}c_w^s\bigg(\frac{L^\alpha}{4n}\bigg)^{-p_w s} e^{\nu'((2L+2L^\alpha+9R)^\beta-(2L)^\beta)} \notag\\
  &\leq \frac{1}{3}+\frac{1}{3}+\frac{1}{3} = 1
  \end{align}
  if $L_1$ was chosen sufficiently large. Note that this choice can be made independently of $\nu'$ (which itself depends on $L_1$). We conclude that \eqref{pf:cor:rescalingimpliesdecay:1} holds for $L\in[2L_1+2L_1^\alpha+9R,2(2L_1)+2(2L_1)^\alpha+9R]$ and hence the inequality is valid in a larger interval as mentioned above.
  \item[(ii)] We proceed very similarly to (i). Choosing $\alpha=1$ and $\nu'= \min\{ \frac{\ln 2 }{(4L_1+9R)^{\gamma^*}}, \frac{\nu_2}{5^{\gamma^*}2},\frac{s\mu_w}{2(20n)^{\gamma^*}}\} $ we have for all $L\in[L_1,4L_1+9R]$:
  \begin{equation}
  B_L \leq 2C'e^{-\nu'L^{\gamma^*}}L^{-q'} \, . \label{pf:cor:rescalingimpliesdecay:2}
  \end{equation} Then for sufficiently large $ L_1 $ and  all $L\in[L_1,2L_1]$ we have
  \begin{align}
  &\frac{(4L+9R)^{q'}e^{\nu'(4L+9R)^{\gamma^*}}}{2C'} B_{4L+9R} \notag  \\
  &\leq 5^{q'}2C'C L^{8dn-q'}e^{\nu'(4L+9R)^{\gamma^*}-2\nu'(2L)^{\gamma^*}} \notag\\
  &\quad + \frac{C}{2C'}5^{q'}L^{q'}e^{\nu'(5L)^{\gamma^*}-\nu_2 L^{\gamma^*}} \notag\\
  &\quad +3^{q'}CL^{6nd}c_w^se^{-s\mu_w\big(\frac{L}{4n}\big)^{\gamma^*}} e^{\nu'((5L)^{\gamma^*}-(2L)^{\gamma^*})} \notag\\
  &\leq \frac{1}{3}+\frac{1}{3}+\frac{1}{3} = 1 .
  \end{align} 
  Hence, \eqref{pf:cor:rescalingimpliesdecay:2} holds for $L\in[4L_1+9R,8L_1+9R]$. Iteration yields the result for all $L\geq L_1$.
\end{itemize}
\end{proof}

\subsection{Proof of Theorem \ref{thm:rescalingineq}}
We divide the proof of Theorem \ref{thm:rescalingineq} into several lemmas. Throughout this section we will use the letter ``$C$'' for a generic positive constant which may change from line to line, but is independent of $L$.

As a first step, we show that it suffices to prove a bound on the quantity $\tilde{B}_s^{(n)}$, which differs from $B_s^{(n)}$ by an additional restriction on the spatial domain $\Omega$. More specifically,
\begin{equation*}
\tilde{B}_s^{(n)}(I,L):=\sup_{\substack{\xbf,\ybf\in\Rdn \\ \dist_H(\xbf,\ybf)\geq L}}\sup_{\substack{\Omega\subset\bigcup_j \Lambda_{3L}(x_j) \\ \mathrm{open,\, bd.}}}\sup_{\substack{\Real z\in I \\ 0<|\Imag z|<1}} \Ex\big[\|\chi_\xbf(H_{\Omega^n}^{(n)}-z)^{-1}\chi_\ybf\|^s\big],
\end{equation*}
which ensures that a volume factor of $|\Omega|$ in the estimates only yields an additional polynomial growth with respect to $L$.

\begin{lemma}\label{lemma:rescalingstepone}
Let $I$ and $\alpha$ be as in Theorem \ref{thm:rescalingineq}. Then for sufficiently large $L$
\begin{equation}
B_s^{(n)}(I,2L+2L^\alpha+9R)\leq CL^{3dn} \tilde{B}_s^{(n)}(I,2L+2L^\alpha+6R).
\end{equation}
\end{lemma}

\begin{proof}
Let $\xbf$, $\ybf$, $\Omega$ and $z$ as in the defnition of $B_s^{(n)}(I,2L+2L^\alpha+9R)$. We assume without loss of generality that $\min_j|x_j-y_1|=\dist_H(\xbf,\ybf)\geq 2L+2L^\alpha+9R$ and define
\begin{equation}
\tilde{\Omega}:=\Omega\cap\bigcup_{j=1}^n\Lambda_{2L+2L^\alpha+8R}(x_j) \ .
\end{equation}
Since $(\supp \chi_\ybf)\cap\tilde{\Omega}^n=\emptyset$ and hence $\chi_\ybf\xi = 0 $ for the cutoff function $ \xi = \xi_{{\Omega}^n, \tilde{\Omega}^n} $ defined in~\eqref{eq:cuttoff}, the geometric resolvent equation~\eqref{eq:geomres} yields
\begin{align}
\|\chi_\xbf(H_{\Omega^n}^{(n)}-z)^{-1}\chi_\ybf\| &\leq \sum_{\wbf\in\Zdn}\|\chi_\xbf(H_{\tilde{\Omega}^n}^{(n)}-z)^{-1}[H_{\Omega^n}^{(n)},\xi]\Theta_\wbf(H_{\Omega^n}^{(n)}-z)^{-1}\chi_\ybf\| \notag \\
&\leq \sum_{\substack{\wbf\in\Zdn \\ \Theta_\wbf\nabla\xi\neq0}} \|\chi_\xbf(H_{\tilde{\Omega}^n}^{(n)}-z)^{-1}\tilde{\Theta}_\wbf(H_{0,\tilde{\Omega}^n}^{(n)}+a)^{\frac{1}{2}}\| \notag\\
&\quad\quad\quad\times\|(H_{0,\tilde{\Omega}^n}^{(n)}+a)^{-\frac{1}{2}}[H_{\Omega^n}^{(n)},\xi]\| \cdot \|\id_{B_1(\wbf)}(H_{\Omega^n}^{(n)}-z)^{-1}\chi_\ybf\| \label{eq:leminserta}
\end{align}
where $(\Theta_\wbf)_{\wbf\in\Zdn}$  is a smooth partition of unity with $ \supp \, \Theta_\wbf \subset B_1(\wbf) $,  $\tilde{\Theta}_\wbf $ is a smooth function with  $ \id_{B_1(\wbf)} \leq  \tilde \Theta_\wbf \leq  \id_{B_2(\wbf)} $ and $a:=1-\inf\sigma(H_0^{(n)})$. The second factor is non-random and bounded. Taking the $s$th moment of a term and conditioning, we arrive at
\begin{align}
&\Ex\big[\|\chi_\xbf(H_{\tilde{\Omega}^n}^{(n)}-z)^{-1}\tilde{\Theta}_\wbf(H_{0,\tilde{\Omega}^n}^{(n)}+a)^{\frac{1}{2}}\|^s\|\id_{B_1(\wbf)}(H_{\Omega^n}^{(n)}-z)^{-1}\chi_\ybf\|^s\big] \notag \\
&\leq  \Ex\big[\Ex\big[\|\chi_\xbf(H_{\tilde{\Omega}^n}^{(n)}-z)^{-1}\tilde{\Theta}_\wbf(H_{0,\tilde{\Omega}^n}^{(n)}+a)^{\frac{1}{2}}\|^{2s}\big|\calF_{\calI_\wbf}\big]^{\frac{1}{2}} \notag \\
& \mkern100mu \times \Ex\big[\|\id_{B_1(\wbf)}(H_{\Omega^n}^{(n)}-z)^{-1}\chi_\ybf\|^{2s}\big|\calF_{\calI_\wbf}\big]^{\frac{1}{2}}\big], \label{proof:resamplingstepone:1}
\end{align}
where $ \calF_{\calI_\wbf} $ is the $ \sigma $-algebra generated by the random variables $ (\eta_\zeta)_{\zeta \not\in \calI_\wbf} $ and the set $\calI_\wbf$ is defined as follows. As $\Theta_\wbf\nabla\xi\neq0$, there is $j\in\{1,\ldots,n\}$ such that $\min_k|w_j-x_k|\geq 2L+2L^\alpha+7R$. We then set
 \begin{equation}
\calI_\wbf:=\{\zeta\in\calI\,|\,\supp U(\cdot-\zeta)\cap( \Lambda_3(y_1)\cup\Lambda_3(w_j))\neq \emptyset\}.
\end{equation}
The crucial point here ist that firstly, $\sum_{\zeta\in\calI_\wbf}N_\zeta\geq 1$ on $B_3(\ybf)$ and $B_3(\wbf)$ and secondly, $\sum_{\zeta\in\calI_\wbf}N_\zeta\equiv 0$ on $B_3(\xbf)$. Lemmas \ref{apriori} and \ref{resampling} hence become applicable and, using $ | B_{2\delta}(\tilde \Omega^n) | \leq C L^{dn} $, we bound the right hand side of \eqref{proof:resamplingstepone:1} up to a multiplicative constant by
\begin{align}
L^{dn}\; \Ex\big[\|\chi_\xbf(\Hhat_{\calI_\wbf,\tilde{\Omega}^n}^{(n)}-z)^{-1}\id_{\Xi_\wbf}\|^s\big]
\end{align}
where $\Xi_\wbf=\{\vbf\in\tilde{\Omega}^n\,|\,\dist(\vbf,\supp\sum_{\zeta\in\calI_\wbf}N_\zeta)<2\}$. The set $\Xi_\wbf$ can be covered by at most $L^{dn}$ balls of radius $1/2$ around configurations $\vbf\in\Zdn$. Furthermore, each of these configurations satisfies $\dist_H(\xbf,\vbf)\geq 2L+2L^\alpha+6R$. Hence,
\begin{multline}
\hat{\Ex}\Ex\big[\|\chi_\xbf(\Hhat_{\calI_\wbf,\tilde{\Omega}^n}^{(n)}-z)^{-1}\id_{\Xi_\wbf}\|^s\big] \\
= \Ex\big[\|\chi_\xbf(H_{\tilde{\Omega}^n}^{(n)}-z)^{-1}\id_{\Xi_\wbf}\|^s\big] \leq L^{dn}\tilde{B}_s^{(n)}(I,2L+2L^\alpha+6R) .
\end{multline}
Combining all of the above, we arrive at
\begin{align}
\Ex\big[\|\chi_\xbf(H_{\tilde{\Omega}^n}^{(n)}-z)^{-1}\chi_\ybf\|^s\big] &= \hat{\Ex}\Ex\big[\|\chi_\xbf(H_{\tilde{\Omega}^n}^{(n)}-z)^{-1}\chi_\ybf\|^s\big] \notag \\
& \leq \sum_{\substack{\wbf\in\Zdn \\ \Theta_\wbf\Theta\neq 0}} C L^{2dn}\tilde{B}_s^{(n)}(I,2L+2L^\alpha+6R) \notag \\
&\leq CL^{3dn} \tilde{B}_s^{(n)}(I,2L+2L^\alpha+6R) ,
\end{align}
which concludes the proof of the lemma.
\end{proof}

We can now proceed with a bound on $\tilde{B}_s^{(n)}$, which yields Theorem \ref{thm:rescalingineq} if it is combined with Lemma \ref{lemma:rescalingstepone}.

\begin{lemma}\label{lemma:rescalingsteptwo}
Let $I$ and $\alpha$ be as in Theorem \ref{thm:rescalingineq}. Then
\begin{multline}
\tilde{B}_s^{(n)}(I,2L+2L^\alpha+6R) \\
\leq C\bigg(L^{5dn}B_s^{(n)}(I,L+L^\alpha)^2+e^{-\nu L^{\alpha\gamma^*}}+ L^{(2+\alpha)nd}\bigg(\bar{w}\bigg(\frac{L^\alpha}{4n}\bigg)\bigg)^s B_s^{(n)}(I,2L)\bigg). \label{eq:rescalingsteptwo:1}
\end{multline}
\end{lemma}

We split the proof of this Lemma into two cases, depending on the diameters of the configurations $\xbf$ and $\ybf$ in $\Ex\big[\|\chi_\xbf(H_{\Omega^n}^{(n)}-z)^{-1}\chi_\ybf\|^s\big]$, with $\xbf$, $\ybf$, $z$ and $\Omega$ as in the definition of $\tilde{B}_s^{(n)}(I,2L+2L^\alpha+6R)$. The diameter is defined as $ \diam(\xbf) = \max_{j,k} | x_j - x_k | $. In the case of one particle, we trivially have $\diam(\xbf)=\diam(\ybf)=0$ and thus Lemma \ref{lemma:rescalingsteptwoparttwo} does not apply. As a consequence, only the first summand appears in the bound \eqref{eq:rescalingsteptwo:1}.

\begin{lemma}\label{lemma:rescalingsteptwopartone}
If $\diam(\xbf),\diam(\ybf)\leq L^\alpha$, then
\begin{equation}
\Ex\big[\|\chi_\xbf(H_{\Omega^n}^{(n)}-z)^{-1}\chi_\ybf\|^s\big] \leq C L^{5dn}B_s^{(n)}(I,L+L^\alpha)^2.
\end{equation}
\end{lemma}

\begin{proof}
Let $j,k$ such that $|x_j-y_k|=\dist_H(\xbf,\ybf)\geq 2L+2L^\alpha+6R$ and define $\Omega_{\xbf}:=\Lambda_{L+L^\alpha+2R}(x_j)\cap\Omega$, $\Omega_{\ybf}:=\Lambda_{L+L^\alpha+2R}(y_k)\cap\Omega$. 
An application of the geometric resolvent inequality \eqref{eq:geomres} with $\Psi_\xbf = \xi_{\Omega^n,\Omega_{\xbf}^n} $ and $\Psi_\ybf = \xi_{\Omega^n,\Omega_{\ybf}^n} $ then yields
\begin{equation}
\chi_\xbf(H_{\Omega^n}^{(n)}-z)^{-1}\chi_\ybf = \chi_\xbf(H_{\Omega_\xbf^n}^{(n)}-z)^{-1}[H_{\Omega^n}^{(n)},\Psi_\xbf](H_{\Omega^n}^{(n)}-z)^{-1}[\Psi_\ybf,H_{\Omega^n}^{(n)}](H_{\Omega_\ybf^n}^{(n)}-z)^{-1}\chi_\ybf
\end{equation}
and hence, proceeding as in~\eqref{eq:leminserta} by inserting the smooth partition of unity $(\Theta_\wbf)_{\wbf\in\Zdn}$ with $\tilde{\Theta}_\wbf $ denoting its enlarged version,
\begin{align}
&\Ex\big[\|\chi_\xbf(H_{\Omega^n}^{(n)}-z)^{-1}\chi_\ybf\|^s] \notag\\
&\leq \sum_{\substack{\wbf,\vbf\in\Zdn \\ \Theta_\wbf\nabla\Psi_\xbf\neq 0 \\ \Theta_\vbf\nabla\Psi_\ybf\neq 0}} \Ex\big[\|\chi_\xbf(H_{\Omega_\xbf^n}^{(n)}-z)^{-1}\tilde{\Theta}_\wbf(H_{0,\Omega_\xbf^n}^{(n)}+a)^{\frac{1}{2}}||^s \|(H_{0,\Omega_\xbf^n}^{(n)}+a)^{-\frac{1}{2}}[H_{\Omega^n}^{(n)},\Psi_\xbf]\|^s \notag\\
&\hphantom{\leq \sum_{\substack{\wbf,\vbf\in\Zdn \\ \Theta_\wbf\nabla\Psi_\xbf\neq 0 \\ \Theta_\vbf\nabla\Psi_\ybf\neq 0}} \Ex\ } \times\|\Theta_{\wbf}(H_{\Omega^n}^{(n)}-z)^{-1}\Theta_\vbf\|^s \|[\Psi_\ybf,H_{\Omega^n}^{(n)}](H_{0,\Omega_\ybf^n}^{(n)}+a)^{-\frac{1}{2}}\|^s \notag\\
&\hphantom{\leq \sum_{\substack{\wbf,\vbf\in\Zdn \\ \Theta_\wbf\nabla\Psi_\xbf\neq 0 \\ \Theta_\vbf\nabla\Psi_\ybf\neq 0}} \Ex\ } \times\|(H_{0,\Omega_\ybf^n}^{(n)}+a)^{\frac{1}{2}}\tilde{\Theta}_\vbf(H_{\Omega_\ybf^n}^{(n)}-z)^{-1}\chi_\ybf\|^s\big] \notag\\
&\leq C\sum_{\substack{\wbf,\vbf\in\Zdn \\ \Theta_\wbf\nabla\Psi_\xbf\neq 0 \\ \Theta_\vbf\nabla\Psi_\ybf\neq 0}} \Ex\big[\Ex\big[\|\chi_\xbf(H_{\Omega_\xbf^n}^{(n)}-z)^{-1}\tilde{\Theta}_\wbf(H_{0,\Omega_\xbf^n}^{(n)}+a)^{\frac{1}{2}}||^{3s} \big| \calF_{\calI_\wbf\cup\calI_\vbf}\big]^{\frac{1}{3}} \notag\\
&\hphantom{\leq C\sum_{\substack{\wbf,\vbf\in\Zdn \\ \Theta_\wbf\nabla\Psi_\xbf\neq 0 \\ \Theta_\vbf\nabla\Psi_\ybf\neq 0}} \Ex\ } \times \Ex\big[\|\Theta_{\wbf}(H_{\Omega^n}^{(n)}-z)^{-1}\Theta_\vbf\|^{3s} \big| \calF_{\calI_\wbf\cup\calI_\vbf}\big]^{\frac{1}{3}} \notag\\
&\hphantom{\leq C\sum_{\substack{ \wbf,\vbf\in\Zdn \\ \Theta_\wbf\nabla\Psi_\xbf\neq 0 \\ \Theta_\vbf\nabla\Psi_\ybf\neq 0}} \Ex\ } \times \Ex\big[\|(H_{0,\Omega_\ybf^n}^{(n)}+a)^{\frac{1}{2}}\tilde{\Theta}_\vbf(H_{\Omega_\ybf^n}^{(n)}-z)^{-1}\chi_\ybf\|^{3s} \big| \calF_{\calI_\wbf\cup\calI_\vbf}\big]^{\frac{1}{3}}\big] , \label{proof:rescalingsteptwopartone:1}
\end{align}
where
\begin{align}
&\calI_\wbf:=\{\zeta\in\Zd \,|\, |\zeta-w_l|\leq 2+r_U\} \, ,  \notag\\
&\calI_\vbf:=\{\zeta\in\Zd \,|\, |\zeta-v_m|\leq 2+r_U\}  \, , 
\end{align}
with $l,m$ chosen such that $\min_i|x_i-w_l|\geq L+L^\alpha+R$ and $\min_i|y_i-v_m|\geq L+L^\alpha+R$. By an application of Lemma \ref{apriori} (bound on fractional moments) and Lemma \ref{resampling} (resampling), each of the summands above can be bounded up to a multiplicative constant by
\begin{align}
&\Ex\big[\Ex\big[\|\chi_\xbf(H_{\Omega_\xbf^n}^{(n)}-z)^{-1}\tilde{\Theta}_\wbf(H_{0,\Omega_\xbf^n}^{(n)}+a)^{\frac{1}{2}}||^{3s} \big| \calF_{\calI_\wbf}\big]^{\frac{1}{3}} \notag\\
&\qquad\times \Ex\big[\|(H_{0,\Omega_\ybf^n}^{(n)}+a)^{\frac{1}{2}} \tilde{\Theta}_\vbf(H_{\Omega_\ybf^n}^{(n)}-z)^{-1}\chi_\ybf\|^{3s} \big| \calF_{\calI_\vbf}\big]^{\frac{1}{3}}\big] \notag\\
&\leq CL^{dn}\, \Ex\Bigg[\mkern-10mu\sum_{\substack{\abf\in\Zdn \\ \id_{B_3(\abf)} \sum_{\zeta\in\calI_\wbf}N_\zeta\neq0}}  \mkern-25mu\|\chi_\xbf(H_{\Omega_\xbf^n}^{(n)}-z)^{-1}\chi_\abf\|^s  \mkern-30mu\sum_{\substack{\bbf\in\Zdn \\ \id_{B_3(\bbf)}\sum_{\zeta\in\calI_\vbf}N_\zeta\neq0}}  \mkern-30mu\|\chi_\bbf(H_{\Omega_\ybf^n}^{(n)}-z)^{-1}\chi_\ybf\|^s\Bigg] \notag\\
&= CL^{dn} \sum_{\substack{ \abf,\bbf\in\Zdn \\ \id_{B_3(\abf)}   \sum_{\zeta\in\calI_\wbf} N_\zeta\neq0 \\ \id_{B_3(\bbf)} \sum_{\zeta\in\calI_\vbf} N_\zeta\neq0}}\mkern-20mu \Ex\big[\|\chi_\xbf(H_{\Omega_\xbf^n}^{(n)}-z)^{-1}\chi_\abf\|^s] \,\Ex\big[\|\chi_\bbf(H_{\Omega_\ybf^n}^{(n)}-z)^{-1}\chi_\ybf\|^s\big].
\end{align}
Note that the potential on $\Omega_\xbf$ and the potential on $\Omega_\ybf$ are stochastically independent, which justifies the last equality. It is easy to see that $\dist_H(\xbf,\abf)\geq L+L^\alpha$ and $\dist_H(\ybf,\bbf)\geq L+L^\alpha$ for any $\abf$ and $\bbf$ appearing in the above sum and hence
\begin{equation}
\Ex\big[\|\chi_\xbf(H_{\Omega_\xbf^n}^{(n)}-z)^{-1}\chi_\abf\|^s]\,  \Ex\big[\|\chi_\bbf(H_{\Omega_\xbf^n}^{(n)}-z)^{-1}\chi_\ybf\|^s\big] \leq B_s^{(n)}(I,L+L^\alpha)^2.
\end{equation}
Taking into account the number of terms in \eqref{proof:rescalingsteptwopartone:1}, we arrive at
\begin{equation}
\Ex\big[\|\chi_\xbf(H_{\Omega^n}^{(n)}-z)^{-1}\chi_\ybf\|^s\big]\leq CL^{5dn}B_s^{(n)}(I,L+L^\alpha)^2
\end{equation}
which completes the proof.
\end{proof}

\begin{lemma}\label{lemma:rescalingsteptwoparttwo}
If $\diam(\xbf)>L^\alpha$ or $\diam(\ybf)>L^\alpha$, then
\begin{equation}
\Ex\big[\|\chi_\xbf(H_{\Omega^n}^{(n)}-z)^{-1}\chi_\ybf\|^s\big] \leq C\bigg(e^{-\nu L^{\alpha\gamma^*}}+  L^{(2+\alpha)nd}\bigg(w_b\bigg(\frac{L^\alpha}{4n}\bigg)\bigg)^s B_s^{(n)}(I,2L)\bigg).
\end{equation}
\end{lemma}

\begin{proof}
Without loss of generality, we can assume that $\diam(\xbf)\geq L^\alpha$. (Here it does not matter whether $\Omega\subset\bigcup_j \Lambda_{3L}(x_j)$ or $\Omega\subset\bigcup_j \Lambda_{3L}(y_j)$ in the definition of $\tilde{B}_s^{(n)}(I,L)$.) Additionally, we choose a partition $J\dot\cup K=\{1,\ldots,n\}$ such that
\begin{equation}
\min_{\substack{j\in J \\ k\in K}}|x_j-x_k|\geq\frac{L^\alpha}{n}, \label{proof:rescalingsteptwoparttwo:eq:1}
\end{equation}
which is possible since $\diam(\xbf)\geq L^\alpha$ (cf.~\cite{AW09}*{Lemma~A.1}), and assume without loss of generality that
\begin{equation}
\min_j|x_1-y_j|\geq 2L+L^\alpha+6R.
\end{equation}
The latter can be justified as follows:  In case there is $j\in\{1.\ldots,n\}$ such that $ \min_k | x_j - y_k| = \dist_H(\xbf,\ybf)=2L+2L^\alpha+6R $ the claim is obvious (for $ x_j $ playing the role of $ x_1 $).  Else there is $j\in\{1.\ldots,n\}$ such that $ \min_k | x_k - y_j| = \dist_H(\xbf,\ybf) $ and we have $|x_1-y_l|\geq |x_1-y_j|-\diam(\ybf) $ for all $ l $.  This yields the claim if $\diam(\ybf)< L^\alpha$. If $\diam(\ybf)\geq L^\alpha$, one may simply interchange the roles of $\xbf$ and $\ybf$ in this case.

An application of the resolvent equation with $W^{(J,K)}_{\Omega^n} = H_{\Omega^n}^{(n)}- H_{\Omega^n}^{(J,K)} $ (cf.~\eqref{eq:HJK}) and \eqref{assumption:locfornonintsubclusters} yields
\begin{align}
\Ex\big[\|\chi_\xbf(H_{\Omega^n}^{(n)}-z)^{-1}\chi_\ybf\|^s\big] &\leq \Ex\big[\|\chi_\xbf(H_{\Omega^n}^{(J,K)}-z)^{-1}\chi_\ybf\|^s\big] \notag\\
&\quad+ \Ex\big[\|\chi_\xbf(H_{\Omega^n}^{(J,K)}-z)^{-1}W^{(J,K)}_{\Omega^n}(H_{\Omega^n}^{(n)}-z)^{-1}\chi_\ybf\|^s\big] \notag \\
& \leq A^*e^{-\nu^*L^{\gamma^*}} + (I) + (II)
\end{align}
where
\begin{align}
(I) &=\mkern-35mu \sum_{\substack{\wbf\in\Zdn \\ \dist_H^{(J,K)}(\xbf,\wbf)\geq  \frac{L^\alpha}{4n}}}\mkern-40mu \Ex\big[\|\chi_\xbf(H_{\Omega^n}^{(J,K)}-z)^{-1}\chi_\wbf\|^{2s}\big]^{\frac{1}{2}} \|W^{(J,K)}_{\Omega^n}\chi_\wbf\|_\infty^s \Ex\big[\|\chi_\wbf(H_{\Omega^n}^{(n)}-z)^{-1}\chi_\ybf\|^{2s}\big]^{\frac{1}{2}}
\end{align}
and
\begin{align}
(II) &= \mkern-35mu \sum_{\substack{\wbf\in\Zdn \\ \dist_H^{(J,K)}(\xbf,\wbf) < \frac{L^\alpha}{4n}}} \mkern-40mu\Ex\big[\|\chi_\xbf(H_{\Omega^n}^{(J,K)}-z)^{-1}\chi_\wbf\|^s \|W^{(J,K)}_{\Omega^n}\chi_\wbf\|_\infty^s \|\chi_\wbf(H_{\Omega^n}^{(n)}-z)^{-1}\chi_\ybf\|^s\big].
\end{align}
As for $(I)$, we use the boundedness of the last two factors and \eqref{assumption:locfornonintsubclusters}. More precisely, the "one-for-all principle" (cf.~\cite{AFSH}*{Lemma~4.2}) allows one to conclude \eqref{assumption:locfornonintsubclusters}  also for $ s \to 2s $ (with an altered decay exponent) such that
\begin{equation}
\Ex\big[\|\chi_\xbf(H_{\Omega^n}^{(J,K)}-z)^{-1}\chi_\wbf\|^{2s}\big] \leq Ce^{-\mu\dist_H^{(J,K)}(\xbf,\wbf)^{\gamma*}} \leq Ce^{-\tilde{\mu}L^{\alpha\gamma*}}  \, . 
\end{equation}
The last step follows from $\dist_H^{(J,K)}(\xbf,\wbf)\geq L^\alpha/{4n}$ and is used 
for  part of the decay only. The other part  controls the sum, and we arrive at
\begin{equation}
(I) \leq Ce^{-\tilde{\nu}L^{\alpha\gamma*}}.
\end{equation}

As for $(II)$, we proceed differently. Let $\wbf\in\Zdn$ such that $\dist_H^{(J,K)}(\xbf,\wbf)<L^\alpha/4n$. Assuming without loss of generality that $1\in J$, we can infer that there exists $j\in J$ such that $|x_1-w_j|< L^\alpha/4n$. As a consequence,
\begin{equation}
\min_l|w_j-y_l| \geq \min_l|x_1-y_l|-|x_1-w_j| \geq 2L+L^\alpha+6R-\frac{L^\alpha}{4n} \geq 2L+6R.
\end{equation}
Now define
\begin{equation}
\calI_{x_1,w_j}=\{\zeta\in\calI\,|\,|x_1-\zeta|\leq r_U+4\text{ or }|w_j-\zeta|\leq r_U+4\}.
\end{equation}
Due to the a priori bound on fractional moments in Lemma \ref{apriori} and the resampling estimate in Lemma \ref{resampling}, we infer that
\begin{align}
&\Ex\big[\|\chi_\xbf(H_{\Omega^n}^{(J,K)}-z)^{-1}\chi_\wbf\|^s\|\chi_\wbf(H_{\Omega^n}^{(n)}-z)^{-1}\chi_\ybf\|^s\big] \notag\\
&\leq \Ex\big[\Ex\big[\|\chi_\xbf(H_{\Omega^n}^{(J,K)}-z)^{-1}\chi_\wbf\|^{2s}\big|\calF_{\calI_{x_1,w_j}}\big]^{\frac{1}{2}}\Ex\big[\|\chi_\wbf(H_{\Omega^n}^{(n)}-z)^{-1}\chi_\ybf\|^{2s}\big|\calF_{\calI_{x_1,w_j}}\big]^{\frac{1}{2}}\big] \notag \\
&\leq C\, \Ex\big[\Ex\big[\|\chi_\wbf(H_{\Omega^n}^{(n)}-z)^{-1}\chi_\ybf\|^{2s}\big|\calF_{\calI_{x_1,w_j}}\big]^{\frac{1}{2}}\big] \notag \\
&\leq CL^{dn}\sum_{\substack{\ubf\in\Zdn \\ \id_{B_3(\ubf)}\sum_{\zeta\in\calI_{x_1,w_j}}N_\zeta\neq0}} \Ex\big[\|\chi_\ubf(H_{\Omega^n}^{(n)}-z)^{-1}\chi_\ybf\|^s\big] \notag  \\
&\leq CL^{2nd} B_s^{(n)}(I,2L).
\end{align}
In the last step we used that the number of contributing summands is of order $L^{dn}$ as $\Omega$ is a subset of $\bigcup_{j=1}^n\Lambda_{3L}(x_j)$. We conclude that
\begin{equation}
(II) \leq \sum_{\substack{\wbf\in\Zdn \\ \chi_\wbf\id_{\Omega^n}\neq 0 \\ \dist_H^{(J,K)}(\xbf,\wbf)<\frac{L^\alpha}{4n}}} CL^{2nd} \, \|W^{(J,K)}\chi_\wbf\|_\infty^s \, B_s^{(n)}(I,2L) \, . 
\end{equation}
For any $\wbf$ in the sum above, there exist $l\in J$ and $m\in K$ such that $\min_{j\in J,k\in K}|w_j-w_k|=|w_l-w_m|$. Additionally, from $\dist_H^{(J,K)}(\xbf,\wbf)\leq L^\alpha/4n$, it follows that there exist $l'\in J$ and $m'\in K$ such that $|w_l-x_{l'}|,|w_m-x_{m'}|\leq L^\alpha/4n$. As a consequence, using 
\begin{equation}
\min_{\substack{j\in J \\ k\in K}}|w_j-w_k| \geq |x_{l'}-x_{m'}| - |x_{l'}-w_l|-|w_m-x_{m'}| \geq \frac{L^\alpha}{2n}
\end{equation}
and hence
\begin{equation}
\|W^{(J,K)}\chi_\wbf\|_\infty^s \leq C \bigg(w_b\bigg(\frac{L^\alpha}{4n}\bigg)\bigg)^s,
\end{equation}
we conclude 
\begin{equation}
(II) \leq C L^{2nd} \, L^{\alpha nd}\bigg(w_b\bigg(\frac{L^\alpha}{4n}\bigg)\bigg)^s B_s^{(n)}(I,2L).
\end{equation}
This finishes the proof.
\end{proof}

Combining the lemmas in this section finally yields the claim of Theorem \ref{thm:rescalingineq}.

\subsection{Initial length scale estimates}\label{section:initialestimates}

We proceed with the proof of the initial length scale estimate that is assumed in Corollary \ref{cor:rescalingimpliesdecay}, i.e., the estimate $B_s^{(n)}(I,L)\leq C'L^{-q'}$. We state two possible situations in the following theorem:

\begin{theorem}\label{thm:initialestimates}
\begin{enumerate}
\item[\textnormal{(i)}] Let $C'>0$ and $q'>8dn$. Suppose $I$ is an interval such that for sufficiently large $L$
\begin{equation}
B_{s,\alpha_W=0}^{(n)}(I,L) \leq \frac{1}{2}C'L^{-q'} \label{eq:thm:initialestimates:1}
\end{equation}
for the non-interacting $n$-particle system, i.e., $\alpha_W=0$. 
Then for any sufficiently large $L_1$ there exists $\alpha_0>0$ such that, if $\alpha_W\in[0,\alpha_0]$, then
\begin{equation}
B_s^{(n)}(I,L)\leq C'L^{-q'}
\end{equation}
for all $L\in[L_1,4L_1+9R]$.
\item[\textnormal{(ii)}] Let $C'>0$ and $q'>8dn$ and suppose that for some $\xi>2(q'+3dn)$ there exists $L^{(n)}>0$ such that
\begin{equation}
\Prob\big(E_0(H_{B_L(\xbf)}^{(n)}) \leq E_0^{(n)}+ L^{-1}\big) \leq L^{-\xi} \label{eq:thm:initalestimates:2}
\end{equation}
for $L\geq L^{(n)}$ and any $\xbf\in\Rdn$. Then for any sufficiently large $L_1$ there exists $\eta^{(n)}>0$ such that
\begin{equation}
B_s^{(n)}(I,L)\leq C'L^{-q'}
\end{equation}
for all $L\in[L_1,4L_1+9R]$, where $I=[E_0^{(n)},E_0^{(n)}+\eta^{(n)}]$.
\end{enumerate}
\end{theorem}

\begin{proof}
The proof of (i) is a perturbation argument. We apply the resolvent equation, use \eqref{eq:thm:initialestimates:1} and obtain
\begin{align}
\Ex\big[\|\chi_\xbf(H_{\Omega^n}^{(n)}-z)^{-1}\chi_\ybf\|^s\big] &\leq \Ex\big[\|\chi_\xbf(H_{\Omega^n,\alpha_W=0}^{(n)}-z)^{-1}\chi_\ybf\|^s\big] \notag\\
&\quad+ \Ex\big[\|\chi_\xbf(H_{\Omega^n,\alpha_W=0}^{(n)}-z)^{-1}\alpha_W  W(H_{\Omega^n}^{(n)}-z)^{-1}\chi_\ybf\|^s\big] \notag\\
&\leq \frac{1}{2}C'L^{-q'} + \alpha_W^s C
\end{align}
for $\dist_H(\xbf,\ybf)\geq L$  with $L $ sufficiently large and $\Omega\subset\Rd$ open, bounded, $\Real z\in I$, $0<|\Imag z|<1$. Given $ L_1 $ sufficiently large  and $\alpha_W^s\leq C'(4L_1+9R)^{-q'}/(2C)$ we can hence conclude
\begin{equation}
B_s^{(n)}(I,L) \leq \frac{1}{2}C'L^{-q'} + \frac{1}{2}C'L^{-q'} = C'L^{-q'}
\end{equation}
for all $L\in[L_1,4L_1+9R]$.

We proceed with the proof of (ii). Let $\xbf\in\Rdn$ and $\Omega_k$, $k\in\{1,\ldots K\}$, be the connected components of $\bigcup_{j=1}^n \Lambda_{3L}(x_j)$. Clearly, $\Omega_k\subset \Lambda_{6nL}(u_k)$ for some $u_k\in\Rd$. Then for any open set $\Omega\subset\bigcup_{j=1}^n \Lambda_{3L}(x_j)$, we have
\begin{align}
& \Prob\big(E_0(H_{\Omega^n}^{(n)}) \leq E_0^{(n)}+(6nL)^{-1}\big) = \Prob\big(\min_{k_1,\ldots,k_n=1}^K E_0(H_{\prod_{j=1}^n\Omega_{k_j}}^{(n)})\leq E_0^{(n)}+(6nL)^{-1}\big) \notag\\
&\leq \sum_{k_1,\ldots,k_n=1}^K\Prob\big( E_0(H_{B_{6nL}(u_{k_1},\ldots,u_{k_n})}^{(n)})\leq E_0^{(n)}+(6nL)^{-1}\big) \notag\\
&\leq n^n(6nL)^{-\xi},
\end{align}
where the first inequality follows from the monotonicity of the ground state energy in the domain and the last inequality follows from \eqref{eq:thm:initalestimates:2} for sufficiently large $L$.
Pick $\ybf\in\Rdn$ such that $\dist_H(\xbf,\ybf)\geq L$ and $z\in\C$ with $\Real z\in [E_0^{(n)},E_0^{(n)}+(12nL)^{-1}]$, $0<|\Imag z|<1$. Then
\begin{align}
& \Ex\big[\|\chi_\xbf(H_{\Omega^n}^{(n)}-z)^{-1}\chi_\ybf\|^s\big] \notag \\
&\leq \Ex\big[\|\chi_\xbf(H_{\Omega^n}^{(n)}-z)^{-1}\chi_\ybf\|^s \id_{\mathbf{\Omega}_B}\big] + \Ex\big[\|\chi_\xbf(H_{\Omega^n}^{(n)}-z)^{-1}\chi_\ybf\|^s \id_{\mathbf{\Omega_B^c}}\big] \notag\\
&\leq \Ex\big[\|\chi_\xbf(H_{\Omega^n}^{(n)}-z)^{-1}\chi_\ybf\|^{2s}\big]^{\frac{1}{2}} \Prob\big(\mathbf{\Omega}_B\big)^{\frac{1}{2}} + \Ex\big[\|\chi_\xbf(H_{\Omega^n}^{(n)}-z)^{-1}\chi_\ybf\|^s \id_{\mathbf{\Omega}_B^c}\big],
\end{align}
where $\mathbf{\Omega}_B:=\{\omega\,|\,E_0(H_{\Omega^n}^{(n)}(\omega))\leq E_0^{(n)}+(6nL)^{-1}\}$. The first term is bounded by $CL^{-\xi/2}$, whereas the second term can be bounded with the help of the Combes-Thomas estimate (cf. \cite{GK03}*{Theorem~1}), since for $\omega\in\mathbf{\Omega}_B^c$ we have $\Real z<E_0(H_{\Omega^n}^{(n)})-(12nL)^{-1}$ and hence
\begin{equation}
\|\chi_\xbf(H_{\Omega^n}^{(n)}(\omega)-z)^{-1}\chi_\ybf\| \leq C Le^{-\mu\sqrt{L}}.
\end{equation}
As a consequence, we have
\begin{equation}
\tilde{B}_s^{(n)}([E_0^{(n)},E_0^{(n)}+(12nL)^{-1}],L) \leq \tilde{c}L^{-\frac{\xi}{2}}
\end{equation}
for sufficiently large $L$. Similarly to Lemma \ref{lemma:rescalingstepone}, we can conclude
\begin{equation}
B_s^{(n)}([E_0^{(n)},E_0^{(n)}+(12 n L)^{-1}],L) \leq cL^{3dn-\frac{\xi}{2}}
\end{equation}
for large $L$. As $\xi>2(q'+3dn)$, this yields the claim of (ii).
\end{proof}

\subsection{Proof of Theorems \ref{thm:main:case1} and \ref{thm:main:case2}}\label{section:proofofmainresults}

\begin{proof}[Proof of Theorem \ref{thm:main:case1}]
\ 
\begin{enumerate}
  \item[(i)] The assumption of Theorem \ref{thm:main:case1}(i) yields the basis of an induction on $n$, i.e.,
\begin{equation}
B_s^{(1)}([E_0^{(1)},E_0^{(1)}+\eta^{(1)}],L) \leq C^{(1)}e^{-\mu^{(1)} L^{\gamma_w}}
\end{equation}
for all $L\geq 0$. Given any $\eta^{(n)}\in(0,\eta^{(1)})$, the induction step proceeds as follows:
Assuming that we have
\begin{equation}
B_s^{(m)}([E_0^{(n-1)},E_0^{(n-1)}+\eta^{(n-1)}],L) \leq C^{(m)}e^{-\mu^{(m)}L^{\gamma_w}}
\end{equation}
for some $\eta^{(n-1)}\in(\eta^{(n)},\eta^{(1)})$ and all $m\leq n-1$, Theorem \ref{thm:fmboundnonint} guarantees the applicability of  Theorem \ref{thm:rescalingineq} with $I=I^{(n)}=[E_0^{(n)},E_0^{(n)}+\eta^{(n)}]$. 
Theorem~\ref{thm:fmboundnonint} also implies the (sub)exponential decay of the fractional moments of the resolvent of the non-interacting system. In particular, the assumption of  Theorem~\ref{thm:initialestimates}(i) is satisfied and we apply Corollary \ref{cor:rescalingimpliesdecay}(ii) to conclude that
\begin{equation}
B_s^{(n)}([E_0^{(n)},E_0^{(n)}+\eta^{(n)}],L) \leq C^{(n)}e^{-\mu^{(n)}L^{\gamma_w}}
\end{equation}
if $\alpha_W$ is sufficiently small. This yields the claim of Theorem \ref{thm:main:case1}(i).
  \item[(ii)] In the case $n=1$ and as a basis of the induction, the assertion follows by  combining  Theorem \ref{thm:initialestimates}(ii) and Corollary \ref{cor:rescalingimpliesdecay}(ii), where we note that the assumption of Theorem \ref{thm:rescalingineq} is trivially satisfied for $ n = 1 $. As in the proof of (i), we proceed with the induction step and assume a bound on $B_s^{(m)}([E_0^{(m)},E_0^{(m)}+\eta^{(n-1)}],L)$ for some $\eta^{(n-1)}\in(0,\eta^{(1)})$ and all $m\leq n-1$. As a consequence, we infer the condition of Theorem \ref{thm:rescalingineq} with $I=[E_0^{(n)},E_0^{(n)}+\eta^*]$ with some $\eta^*<\eta^{(n-1)}$. If we choose $C'=1$ and $q'>8dn$ such that $\xi=\xi^{(n)}>2(q'+3dn)$ (which is possible due to the assumption on $\xi^{(n)}$), Theorem \ref{thm:initialestimates}(ii) ensures the existence of an $\eta^{(n)}\leq \eta^*$ such that Corollary \ref{cor:rescalingimpliesdecay}(ii) can be applied with $I=I^{(n)}=[E_0^{(n)},E_0^{(n)}+\eta^{(n)}]$, which concludes the proof.
\end{enumerate}
\end{proof}

\begin{proof}[Proof of Theorem \ref{thm:main:case2}]
The proof proceeds in the same fashion as the proof of Theorem \ref{thm:main:case1}. Therefore, we focus on the differences only.

We conduct the induction step as long as $n \leq p_w s/12= p_w /48$  with $s=1/4$. The $\alpha=\alpha^{(n)}$ in the rescaling inequality in Theorem \ref{thm:rescalingineq} is chosen such that the resulting $\beta=\beta^{(n)} = \min\{\alpha\gamma^*  ,  1-\alpha  \} $ with $\gamma^* = \beta^{(n-1)}$  in Corollary \ref{cor:rescalingimpliesdecay}(i) is maximal. This is the case if $\alpha\gamma^* = 1-\alpha$, i.e., $\alpha^{(n)}=(1+\beta^{(n-1)})^{-1}$ and hence $\beta^{(n)}=\beta^{(n-1)}/(1+\beta^{(n-1)}) $. The induction basis holds with $\beta^{(1)}\in(0,1]$ and consequently, for $ n \geq 2 $ we have $\alpha^{(n)}=(1+(n-2)\beta^{(1)})/(1+(n-1)\beta^{(1)})$ and $\beta^{(n)}=\beta^{(1)}/(1+(n-1)\beta^{(1)})$. In particular, $\alpha^{(n)}\geq 1/2$, so our condition on $n$ ensures the applicability of Corollary~\ref{cor:rescalingimpliesdecay}(i). Apart from these considerations, the proof proceeds analogously to the proof of Theorem \ref{thm:main:case1}.
\end{proof}

\appendix

\section{A priori estimates}

In this section, we present some auxiliary results that are essential for our analysis of the multiparticle system. Most of these are familiar from the one-particle case (cf. \cite{AENSS}). We state the results in a more general setting:
\begin{itemize}
\item The random operator $H(\omega)=H_0+W+V(\omega)=-\Delta+W+V(\omega)$ acts on $L^2(\RD)$, where $D\in\N$ is arbitrary.
\item The potential $W$ is not necessarily an interaction potential, but can be an arbitrary bounded background potential.
\item The random potential $V(\omega)$ takes the form
\begin{equation}
V(\omega)=\sum_{\zeta\in\calI}\eta_\zeta(\omega) N_\zeta,
\end{equation}
where $\calI$ is an arbitrary countable index set, the functions $N_\zeta$ are measurable, non-negative and satisfy
\begin{equation}
\sup_{\xbf\in\RD}\sum_{\zeta\in\calI}N_\zeta<\infty,
\end{equation}
and the random variables $\eta_\zeta$, $\zeta\in\calI$, are independent and identically distributed. The distribution of $\eta_\zeta$ has a bounded density $\rho$ with a compact support that is a subset of the non-negative half-line $\R_0^+$.
\end{itemize}
In addition, we introduce some notation. For any set $\calI'\subset\calI$, $\calF_{\calI'}$ is the $\sigma$-algebra generated by the random variables $(\eta_\zeta)_{\zeta\in\calI\backslash\calI'}$. The infimum of the spectrum of $H_w=-\Delta+W$ is denoted by $E_0$.

The first lemma is basically a reformulation of \cite{AENSS}*{Lemma 3.3} in the above setting. Its message is the following: Suppose $\Lambda_1',\Lambda_2'\subset\RD$ are bounded open sets and $\calI'\subset\calI$ such that the sum $\sum_{\zeta\in\calI'}N_\zeta$ is bounded from below by a positive number on a neighborhood of these sets. Then it suffices to average over the ``local'' random variables $\eta_\zeta$, $\zeta\in\calI'$, in order that a fractional moment of $\|\id_{\Lambda_1'}(H-z)^{-1}\id_{\Lambda_2'}\|$, $z\in\C\setminus\R$, to be bounded uniformly with respect to $\Imag z$.

\begin{lemma}\label{apriori}
Let $\Lambda_1,\Lambda_2\subset\RD$ be open and bounded sets and let $\calI_1,\calI_2\subset\calI$ be finite sets and $c_0\in(0,\infty)$ such that $\inf_{\xbf\in\Lambda_j}\sum_{\zeta\in\calI_j}N_\zeta\geq c_0$ for both $j=1$ and $j=2$. Suppose $\Lambda_j'\subset\Lambda_j$ such that $\dist(\partial\Lambda_j,\Lambda_j')\geq\delta>0$. Then there exists $C>0$ such that for any open set $\Omega\subset\RD$, $s\in(0,1)$ and $z\in\C\backslash\R$
\begin{multline}
\Ex\big[\|\id_{\Lambda_1'}(H_\Omega-z)^{-1}\id_{\Lambda_2'}\|^s\big|\calF_{\calI_1\cup\calI_2}\big] \\
\leq \frac{C^s}{1-s}(1+|\Lambda_1|^{\frac{1}{2}}|\Lambda_2|^{\frac{1}{2}})^s (\#(\calI_1\cup\calI_2))^{(2D+8)s}
\bigg(1+|z-E_0|+\delta^{-2}\bigg)^{(D+3)s}.
\end{multline}
\end{lemma}

\begin{proof}
We only give a sketch of the proof, as it does not differ much from the one found in \cite{AENSS}.
We pick the least integer $m>D/2$ and choose a family of smooth cutoff functions 
$(\Theta_l)_{l\in\{1,\ldots,m+1\}}$ satisfying $0\leq\Theta_l\leq 1$, $\Theta_{l+1}\equiv 1$ on $\supp\Theta_l$ for $l\leq m$, $\Theta_1\equiv 1$ on $\Lambda_1'$ and $\supp\Theta_{m+1}\subset\Lambda_1$. Define the operator $H'=H-\sum_{\zeta\in\calI_1\cup\calI_2}\eta_\zeta N_\zeta$ and let $a=1-E_0$. By a repeated application of the identity
\begin{align}
\Theta_j^2(H-z)^{-1} &= \Theta_j(H'+a)^{-1}\Theta_j \notag\\
&\quad -\sum_{\zeta\in\calI_1\cup\calI_2}\eta_\zeta\Theta_j(H'+a)^{-1}\Theta_j N_\zeta\Theta_2^2(H-z)^{-1} \notag\\
&\quad+\Theta_j(H'+a)^{-1}(\Theta_j(a+z)+[\Theta_j,\Delta])\Theta_{j+1}^2(H-z)^{-1} ,
\end{align}
one can see that
\begin{equation}
\Theta_1^2(H-z)^{-1} = B_1 + T_1\Theta_{m+1}^2(H-z)^{-1}, \label{thetaonesqres}
\end{equation}
where $B_1$ is a bounded operator and $T_1$ is a Hilbert-Schmidt operator.

Proceeding similarly to the above, we can show that
\begin{equation}
(H-z)^{-1}\Psi_1^2 = B_2 + (H-z)^{-1}\Psi_{m+1}^2T_2  \label{eq:respsionesq}
\end{equation}
where $\Psi_1$, $\Psi_{m+1}$, $B_2$ and $T_2$ are analogous to $\Theta_1$, $\Theta_{m+1}$, $B_1$ and $T_1$, respectively, but with $\Lambda_2'$ and $\Lambda_2$ playing the roles of $\Lambda_1'$ and $\Lambda_1$, respectively. Combining this with \eqref{thetaonesqres} yields
\begin{equation}
\Theta_1^2(H-z)^{-1}\Psi_1^2 = \bar{B} +T_1\Theta_{m+1}^2(H-z)^{-1}\Psi_{m+1}^2T_2 \label{t1resolvt2}
\end{equation}
with $\bar{B}=B_1\Psi_1^2+T_1\Theta_{m+1}^2B_2$.

The operators $T_j$, $j\in\{1,2\}$ are polynomials in the random variables $(\eta_\zeta)_{\zeta\in\calI_1\cup\calI_2}$. They can be written as
\begin{equation}
T_j= \sum_{\substack{\alpha\in\N_0^{\calI_1\cup\calI_2} \\ |\alpha|\leq m}}T_{j,\alpha}\prod_{\zeta\in\calI_1\cup\calI_2}\eta_\zeta^{\alpha_\zeta} \label{formoftj}
\end{equation}
where $|\alpha|=\sum_{\zeta}\alpha_\zeta$ is the absolute value of the multiindex $\alpha$ and each operator $T_{j,\alpha}$ is Hilbert-Schmidt and independent of $(\eta_\zeta)_{\zeta\in\calI_1\cup\calI_2}$.

Define $F_j=\sum_{\zeta\in\calI_j}N_\zeta$ for $j\in\{1,2\}$. As $F_j$ is strictly positve on $\Lambda_j$, we can insert $F_1/F_1$ and $F_2/F_2$ next to $\Theta_{m+1}$ and $\Psi_{m+1}$, respectively, into \eqref{t1resolvt2} and obtain
\begin{equation}
\Theta_1^2(H-z)^{-1}\Psi_1^2 = \bar{B} +\sum_{\substack{\zeta\in\calI_1 \\ \gamma\in\calI_2}}T_1\Theta_{m+1}^2\frac{N_\zeta}{F_2}(H-z)^{-1}\frac{N_\gamma}{F_2}\Psi_{m+1}^2T_2 \,.
\end{equation}
Using this and \eqref{formoftj}, we estimate
\begin{align}
&\Prob\big(\|\id_{\Lambda_1'}(H-z)^{-1}\id_{\Lambda_2'}\|>t\big|\calF_{\calI_1\cup\calI_2}\big) \leq \Prob\big(\|\Theta_1^2(H-z)^{-1}\Theta_2^2\|>t\big|\calF_{\calI_1\cup\calI_2}\big)  \notag\\
&\leq \Prob\bigg(\|\bar{B}\|>\frac{t}{M}\bigg|\calF_{\calI_1\cup\calI_2}\bigg) \notag\\
&\quad+c\sum_{\substack{\zeta\in\calI_1 \\ \gamma\in\calI_2}}\sum_{|\alpha|,|\beta|\leq m}\Prob\bigg(\bigg\|T_{1,\alpha}\Theta_{m+1}^2\frac{N_\zeta}{F_2}(H-z)^{-1}\frac{N_\gamma}{F_2}\Psi_{m+1}^2T_{2,\beta}\bigg\|>\frac{t}{M}\bigg|\calF_{\calI_1\cup\calI_2}\bigg) \,,  \label{weakl1bbarandresolvent}
\end{align}
where $M$ is the total number of terms in \eqref{weakl1bbarandresolvent}. As $\bar{B}$ is uniformly bounded, the first term can be bounded by $C_B/t$ for a suitable $C_B>0$.
As for the other terms in \eqref{weakl1bbarandresolvent}, we employ Proposition \ref{weakl1estimate} below. In order to estimate a term from above, we first integrate over the random variables $\eta_\zeta$ and $\eta_\gamma$. For this purpose, we employ \eqref{weakl1estimate1} or \eqref{weakl1estimate2}, depending on whether $\zeta$ is equal to $\gamma$. This yields a weak $L^1$-bound as above.

We can therefore conclude
\begin{equation}
\Prob\big(\|\id_{\Lambda_1'}(H-z)^{-1}\id_{\Lambda_2'}\|>t\big|\calF_{\calI_1\cup\calI_2}\big) \leq \frac{C_B+C_T}{t}
\end{equation}
for suitable constants $C_B,C_T< \infty $. This implies the claim of the lemma, since for a random variable $X$ the weak $L^1$-estimate $\Prob(|X|>t|...)\leq C/t$ implies $\Ex[|X|^s|...] \leq C^s/(1-s)$. The specific form of the claimed bound can be obtained by a careful analysis of the proof. We omit this analysis as it is mostly just a repetition of arguments in the proof of \cite{AENSS}*{Lemma 3.3}.
\end{proof}

The proof above used the following proposition, which was proved in \cite{AENSS}. (Actually, only the second inequality of the proposition was stated in \cite{AENSS}. However, the proof of the first one is more or less contained in the proof of the second one.)
\begin{proposition}[Proposition 3.2 in \cite{AENSS}]\label{weakl1estimate}
Suppose $\mathcal{H}$ is a separable Hilbert space and $A$ a maximally dissipative operator on $\mathcal{H}$ with strictly positive imaginary part. If $M_1,M_2$ are Hilbert-Schmidt operators and $U_1,U_2$ non-negative operators, then for all $\lambda,t>0$
\begin{equation}
\bigg|\bigg\{v\in[0,\lambda]\bigg|\,\bigg\|M_1U_1^{\frac{1}{2}}(A-vU_1)^{-1}U_1^{\frac{1}{2}}M_2\bigg\|_{\mathrm{HS}}>t\bigg\}\bigg| \leq C_w\|M_1\|_{\mathrm{HS}}\|M_2\|_{\mathrm{HS}}\frac{1}{t} \label{weakl1estimate1}
\end{equation}
and
\begin{multline}
\bigg|\bigg\{(v_1,v_2)\in[0,\lambda]^2\bigg|\,\bigg\|M_1U_1^{\frac{1}{2}}(A-v_1U_1-v_2U_2)^{-1}U_2^{\frac{1}{2}}M_2\bigg\|_{\mathrm{HS}}>t\bigg\}\bigg| \\
\leq 2C_w\lambda\|M_1\|_{\mathrm{HS}}\|M_2\|_{\mathrm{HS}}\frac{1}{t}. \label{weakl1estimate2}
\end{multline}
\end{proposition}

The next lemma is a reformulation of \cite{AENSS}*{Lemma 3.4}. The setting is similar to Lemma \ref{apriori}. A resolvent $(H-z)^{-1}$ is again bracketed between two functions supported on bounded sets. However, in this case, we require that $\sum_{\zeta\in\calI'}N_\zeta$ covers only a neighborhood of one of these sets. As a result, averaging over the random variables $\eta_\zeta$, $\zeta\in\calI'$, yields a bound that includes the resolvent of an operator $\Hhat$, which is the Hamiltonian $H$ with ``resampled'' random variables $\hat{\eta}_\zeta$, $\zeta\in\calI'$. In order to state this more rigorously, we introduce some new notation: Let $\hat{\eta}_\zeta$, $\zeta\in\calI$, be a collection of random variables with the same distribution as (but independent of) the random variables $\eta_\zeta$, $\zeta\in\calI$. For any $\calI'\subset\calI$, we denote by $\Hhat_{\calI'}$ the resampled operator $\Hhat_{\calI'}=H-\sum_{\zeta\in\calI'}\eta_\zeta N_\zeta+\sum_{\zeta\in\calI'}\hat{\eta}_\zeta N_\zeta$.

\begin{lemma}\label{resampling}
Let $\Lambda_1,\Lambda_2\subset\RD$ be disjoint, open and bounded sets, and let $\calI'\subset\calI$ be a finite set and $c_0\in(0,\infty)$ such that $\sum_{\zeta\in\calI'}N_\zeta\geq c_0$ on $\Lambda_2$ and $\sum_{\zeta\in\calI'}N_\zeta \equiv 0$ on $B_{2\delta}(\Lambda_1)$. Suppose $0\leq\Psi\leq1$ is a smooth function supported in an open set $\Lambda_2'\subset\Lambda_2$ satisfying $\dist(\partial\Lambda_2,\Lambda_2')\geq 2\delta$ for some $\delta>0$. Then for all $s\in(0,1)$, there exists $C>0$ such that for all open and bounded sets $\Omega\subset\RD$
\begin{align}
&\max\left\{\Ex\big[\|\id_{\Lambda_1}(H_\Omega-z)^{-1}\id_{\Lambda_2'}\|^s\big|\calF_{\calI'}\big],\Ex\big[\|\id_{\Lambda_1}(H_\Omega-z)^{-1}\Psi(H_{w,\Omega}+a)^{1/2}\|^s\big|\calF_{\calI'}\big]\right\} \notag\\
&\leq \frac{1}{1-s}C^s(\#\calI')^{(3D+10)s}(1+|B_{2\delta}(\Omega\cap\supp{\textstyle \sum_{\zeta\in\calI'}N_\zeta})|)^s(1+|z-E_0|+\delta^{-2})^{(D+3)s} \notag\\ &\qquad\times\|\id_{\Lambda_1}(\Hhat_{\calI',\Omega}-z)^{-1}\id_\Xi\|^s, \label{eq:resampling}
\end{align}
where $a=1-E_0$ and $\Xi=\{\xbf\in\RD\,|\,\dist(\xbf,\supp\sum_{\zeta\in\calI'}N_\zeta)<2\delta\}$.
\end{lemma}
The proof follows closely that of \cite{AENSS}*{Lemma 3.4}.

\begin{proof}
Due to the assumption that $\Lambda_1\cap\Lambda_2=\emptyset$, we have $\id_{\Lambda_1}\Psi=0$ and thus
\begin{equation}
\id_{\Lambda_1}(H-z)^{-1}\Psi(H_w+a)^{1/2} = \id_{\Lambda_1}(H-z)^{-1}\tilde{\Psi}^2((a+z-V)\Psi+[\Psi,H_w])(H_w+a)^{-1/2} \ ,
\end{equation}
where $\tilde{\Psi}$ is a suitable smooth function with $0\leq\tilde{\Psi}\leq 1$, $\tilde{\Psi}\equiv 1$ on $\Lambda_2'$ and $\supp\tilde{\Psi}\subset B_\delta(\Lambda_2')$. This yields the estimate
\begin{equation}
\|\id_{\Lambda_1}(H-z)^{-1}\Psi(H_w+a)^{1/2}\|\leq C(1+|z-E_0|+\delta^{-2})\|\id_{\Lambda_1}(H-z)^{-1}\tilde{\Psi}^2\|
\end{equation}
with a suitable constant $C$. In order to prove \eqref{eq:resampling}, we thus need to find a bound on $\Ex[\|\id_{\Lambda_1}(H-z)^{-1}\tilde{\Psi}^2\|^s|\calF_{\calI'}]$. An application of the resolvent equality yields
\begin{align}
\id_{\Lambda_1}(H-z)^{-1}\tilde{\Psi}^2 
  &=\id_{\Lambda_1}(\Hhat_{\calI'}-z)^{-1}\tilde{\Psi}^2 + \id_{\Lambda_1}(\Hhat_{\calI'}-z)^{-1}(\Hhat_{\calI'}-H)(H-z)^{-1}\tilde{\Psi}^2 \ .
\end{align}
Applying an identity of the same form as \eqref{eq:respsionesq}, we get
\begin{equation}
(H-z)^{-1}\tilde{\Psi}^2=\widehat{\Psi}^2\widehat{B}+(H-z)^{-1}\widehat{\Psi}^2\widehat{T} \ ,
\end{equation}
where $\widehat{B}$ is bounded, $\widehat{T}$ is Hilbert-Schmidt and $\widehat{\Psi}$ is a smooth function such that $0\leq\widehat{\Psi}\leq 1$, $\widehat{\Psi}\equiv 1$ on $\supp\tilde{\Psi}$ and $\supp\widehat{\Psi}\subset \Lambda_2$. Therefore,
\begin{align}
\id_{\Lambda_1}(H-z)^{-1}\tilde{\Psi}^2
  &=\id_{\Lambda_1}(\Hhat_{\calI'}-z)^{-1}\tilde{\Psi}^2 + \id_{\Lambda_1}(\Hhat_{\calI'}-z)^{-1}\widehat{\Psi}^2(\Hhat_{\calI'}-H)\widehat{B} \notag\\
  &\quad+ \id_{\Lambda_1}(\Hhat_{\calI'}-z)^{-1}\Theta^2(\Hhat_{\calI'}-H)(H-z)^{-1}\widehat{\Psi}^2\widehat{T} ,
\end{align}
where $\Theta$ is a smooth cutoff function satisfying $\Theta\equiv1$ on $\supp\sum_{\zeta\in\calI'}N_\zeta$ and $\Theta(\xbf)=0$ if $\dist(\xbf,\supp\sum_{\zeta\in\calI'}N_\zeta)>\delta$.
Next we apply the commutator argument to  $(\Hhat_{\calI'}-z)^{-1}\Theta^2$ and obtain
\begin{equation}
(\Hhat_{\calI'}-z)^{-1}\Theta^2=\tilde{\Theta}^2 B+(\Hhat_{\calI'}-z)^{-1}\tilde{\Theta}^2 T
\end{equation}
with a Hilbert-Schmidt operator $T$, a smooth function $\tilde{\Theta}$ with $\supp\tilde{\Theta}\subset B_\delta(\supp\Theta)$ and a bounded operator $B$.  Combining these identities and using that $\Lambda_1\cap\supp\tilde{\Theta}=\emptyset$, we conclude
\begin{align}
\id_{\Lambda_1}(H-z)^{-1}\tilde{\Psi}^2 &= \id_{\Lambda_1}(\Hhat_{\calI'}-z)^{-1}\tilde{\Psi}^2 + \id_{\Lambda_1}(\Hhat_{\calI'}-z)^{-1}\widehat{\Psi}^2(\Hhat_{\calI'}-H)\widehat{B} \notag\\
  &\quad +\id_{\Lambda_1}(\Hhat_{\calI'}-z)^{-1}\tilde{\Theta}^2 T(\Hhat_{\calI'}-H)(H-z)^{-1}\widehat{\Psi}^2\widehat{T}
\end{align}
and hence
\begin{align}
\|\id_{\Lambda_1}(H-z)^{-1}\tilde{\Psi}^2\| &\leq C_0(\#\calI')^{D/2+1}(1+|z-E_0|+\delta^{-2})^{D/2+1}\| \id_{\Lambda_1}(\Hhat_{\calI'}-z)^{-1}\widehat{\Psi}^2\| \notag\\
  &\quad +\|\id_{\Lambda_1}(\Hhat_{\calI'}-z)^{-1}\tilde{\Theta}\|\cdot\| T(\Hhat_{\calI'}-H)(H-z)^{-1}\widehat{\Psi}^2\widehat{T}\|  \label{resamp1} \notag\\
  &\leq C_0(\#\calI')^{D/2+1}(1+|z-E_0|+\delta^{-2})^{D/2+1}\| \id_{\Lambda_1}(\Hhat_{\calI'}-z)^{-1}\widehat{\Psi}^2\| \notag\\
  &\quad +\|\id_{\Lambda_1}(\Hhat_{\calI'}-z)^{-1}\tilde{\Theta}^2\|\sum_{\zeta,\beta\in\calI'}|\hat{\eta}_\zeta-\eta_\zeta|\| TN_\zeta(H-z)^{-1}\frac{N_\beta}{F}\widehat{\Psi}^2\widehat{T}\| ,
\end{align}
where $F=\sum_{\beta\in\calI'}N_\beta$. Taking the $s^{\text{th}}$ power and averaging over the random variables $\eta_\zeta$, $\zeta\in\calI'$, as in the proof of Lemma \ref{apriori} now yields
\begin{align}
&\Ex\big[\|\id_{\Lambda_1}(H-z)^{-1}\tilde{\Psi}^2\|^s\big|\calF_{\calI'}\big] \notag\\
&\leq C_0^s(\#\calI')^{(D/2+1)s}(1+|z-E_0|+\delta^{-2})^{(D/2+1)s}\| \id_{\Lambda_1}(\Hhat_{\calI'}-z)^{-1}\widehat{\Psi}^2\|^s \notag\\
&\quad+\frac{1}{1-s}C_1^s(\#\calI')^{(3D+10)s}(1+|B_{2\delta}(\Omega\cap\supp{\textstyle \sum_{\zeta\in\calI'}N_\zeta})|)^s(1+|z-E_0|+\delta^{-2})^{(D+2)s} \notag\\ &\qquad\times\|\id_{\Lambda_1}(\Hhat_{\calI'}-z)^{-1}\tilde{\Theta}\|^s
\end{align}
\end{proof}

In our analysis in Section \ref{section:fmec} the assumptions (W) and (GW)(iii) take the form of a Wegner estimate. The following lemma ensures that these assumptions are indeed valid in our model. In comparison to the Wegner estimates in \cites{BCSS,HK13,Ki08,KlZe09}, the important feature of the lemma is that it requires only a ``local'' average, similarly to the apriori bound on fractional moments of the resolvent in Lemma \ref{apriori}. More specifically, given a bounded set $\Lambda'\subset\RD$, we average only over a collection of random variables $\eta_\zeta$ with the property that the corresponding potentials $N_\zeta$ cover a neighborhood of $\Lambda'$. The lemma shows that in this case the conditional expectation of the trace of $\id_{\Lambda'}P_I(H)$ is bounded by a constant times the Lebesgue measure of $I$, for any Borel set $I\subset\R$ whose supremum does not exceed a fixed bound.

\begin{lemma}\label{lemma:exptrproj}
Let $\Lambda\subset\RD$ be open and bounded and let $\calI'\subset\calI$ be a finite set such that $\inf_{\xbf\in\Lambda}\sum_{\zeta\in\calI'}N_\zeta(\xbf)\geq c_0$ for some $c_0\in(0,\infty)$. Suppose $\Lambda'\subset\Lambda$ is an open set such that $\dist(\partial\Lambda,\Lambda')\geq\delta>0$. Then there exists $C<\infty$ such that for any open set $\Omega\subset\RD$
\begin{equation}
\Ex\left[\Tr(\id_{\Lambda'} P_I(H_{\Omega}))\big| \calF_{\calI'}\right]\leq C(\#\calI')^{D+4} \sup_{E\in I}|E+a|^{D+2}(1+\delta^{-2})^{D+2}|\Lambda||I|.  \label{exptrprojineq}
\end{equation}
\end{lemma}

\begin{proof}
Choose a sequence of functions $\Theta_j$ as in the proof of Lemma \ref{apriori}, satisfying $\Theta_1\equiv1$ on $\Lambda'$ and $\supp\Theta_{m+1}\subset\Lambda$. Let $a=1-E_0$ and define $H'=H-\sum_{\zeta\in\calI'}\eta_\zeta N_\zeta$. Then for all $j$, we have
\begin{align}
P_I(H)\Theta_j^2 &= P_I(H)(H'+a)\Theta_j\frac{1}{H'+a}\Theta_j +P_I(H)[\Theta_j,H']\frac{1}{H'+a}\Theta_j \notag\\
&= (H+a)P_I(H)\Theta_{j+1}^2\Theta_j\frac{1}{H'+a}\Theta_j -\sum_{\zeta\in\calI'} P_I(H)\Theta_{j+1}^2 \eta_\zeta N_\zeta\Theta_j\frac{1}{H'+a}\Theta_j \notag\\
&\quad +P_I(H)\Theta_{j+1}^2[\Theta_j,H']\frac{1}{H'+a}\Theta_j .
\end{align}
By induction, we obtain
\begin{equation}
P_I(H)\Theta_1^2= \sum_{l} \big(\prod_{\zeta\in\mathcal{I}'}\eta_\zeta^{K_\zeta^{(l)}}\big)  (H+a)^{k^{(l)}}P_I(H)\Theta_{m+1}^2 T^{(l)}  \ ,
\end{equation}
where $K^{(l)}\in \N_0^{\mathcal{I}'}$, $k^{(l)}\in\N_0$, $|K^{(l)}|,k^{(l)}\leq m$, and $T^{(l)}$ is a product of $m$ terms of the form
\begin{equation}
\Theta_j\frac{1}{H'+a}\Theta_j,\ -N_\zeta\Theta_j\frac{1}{H'+a}\Theta_j\text{ or } [\Theta_j,H']\frac{1}{H'+a}\Theta_j \ .
\end{equation}
Each of these operators is in the Schatten class $\mathcal{J}_p$ for any $p>D$. Since $m>D/2$, $T^{(l)}$ is a Hilbert-Schmidt operator. As a consequence, we have
\begin{align}
\Tr(\id_{\Lambda'} P_I(H)) &\leq \|\id_{\Lambda'} P_I(H)\Theta_1^2\|_1 \notag\\
&\leq c\sum_l \|\id_{\Lambda'} P_I(H)(H+a)^{k^{(l)}}P_I(H)\Theta_{m+1}^2T^{(l)} \|_1
\end{align}
Due to the assumption on $\calI'$, we can insert  $\sum_{\zeta\in\calI'}N_\zeta/F$ with $F:=\sum_{\zeta\in\calI'} N_\zeta \geq c_0$ next to $\Theta_{m+1}^2$ and obtain
\begin{align}
&\Tr(\id_{\Lambda'} P_I(H)) \leq c \sum_l\sum_{\zeta\in\calI'} \big\|\id_{\Lambda} P_I(H)(H+a)^{k^{(l)}}P_I(H) \frac{N_\zeta}{F}\Theta_{m+1}^2T^{(l)}\big\|_1 \notag\\
&\leq c \sum_l\sum_{\zeta\in\calI'} \sqrt{\Tr(\id_{\Lambda'} P_I(H)(H+a)^{2k^{(l)}}\id_{\Lambda'})} \sqrt{\Tr\big(T^{(l)^*}\Theta_{m+1}^2\frac{N_\zeta}{F}P_I(H)\frac{N_\zeta}{F}\Theta_{m+1}^2 T^{(l)}\big)}  \notag\\
&\leq c\sup_{E\in I}|E+a|^m\sqrt{\Tr(\id_{\Lambda'} P_I(H))}\sum_l\sum_{\zeta\in\calI'} 
\sqrt{\Tr\big(T^{(l)^*}\Theta_{m+1}^2\frac{N_\zeta}{F}P_I(H)\frac{N_\zeta}{F}\Theta_{m+1}^2 T^{(l)}\big)} \ .
\end{align}
Taking the expected value, conditioned on $\calF_{\calI'}$, and using H\"older's inequality, we get
\begin{multline}
\Ex\big[\Tr(\id_{\Lambda'} P_I(H))\big| \calF_{\calI'}\big] \leq c\sup_{E\in I}|E+a|^m \sum_l\sum_{\zeta\in\calI'}\sqrt{\Ex\big[\Tr(\id_{\Lambda'} P_I(H))\big| \calF_{\calI'}\big]}  \\ \times\sqrt{\Ex\big[\Tr\big(T^{(l)^*}\Theta_{m+1}^2\frac{N_\zeta}{F}P_I(H)\frac{N_\zeta}{F}\Theta_{m+1}^2 T^{(l)}\big)\big| \calF_{\calI'}\big]} \ .
\end{multline}
Let $\sum_j \kappa_j \langle\psi_j,\cdot\rangle\phi_j$ be the singular value decomposition of $T^{(l)^*}\Theta_{m+1}^2\sqrt{N_\zeta}/F$. Then it follows that
\begin{multline}
\Ex\big[\Tr\big(T^{(l)^*}\Theta_{m+1}^2\frac{N_\zeta}{F}P_I(H)\frac{N_\zeta}{F}\Theta_{m+1}^2 T^{(l)}\big)\big|\calF_{\calI'}\big] \\
= \sum_j \kappa_j^2 \Ex\big[ \langle\psi_j,\sqrt{N_\zeta}P_I(H)\sqrt{N_\zeta}\psi_j\rangle \big|\calF_{\calI'}\big] 
\leq \|\rho\|_\infty\|T^{(l)}\|_\mathrm{HS}^2|I| \ .
\end{multline}
The last step is obtained by first integrating out the random variable $\eta_\zeta$, where we apply the spectral averaging lemma from \cite{CH94}. The estimates above now allow us to conclude the claim of the lemma. The specific form of the bound is obtained similarly to Lemmas \ref{apriori} and \ref{resampling} by explicit estimates on the Hilbert-Schmidt norms.
\end{proof}

\section{Decay estimates for an energetically restricted resolvent}\label{appendix:hsct}

In this appendix, we study the decay of ``operator kernels'' of a certain type of functions of a self-adjoint operator $H$ that satisfies a Combes-Thomas estimate. More specifically, we analyze how fast a bound on a term of the form
\begin{equation*}
\|\chi_\xbf f(H)\chi_\ybf\|
\end{equation*}
decays with respect to the distance of $\xbf$ and $\ybf$. In \cites{GK03,BGK} the Helffer-Sj\"ostrand formula (cf. \cites{HS,Da}) and a Combes-Thomas estimate (cf. \cites{CT73,GK03}) were used to study the decay of operator kernels in the case of  \textit{$L^1$-Gevrey fuctions} $ f $. Among other things, the results obtained there yield sub-exponential decay of $\|\chi_\xbf \zeta(H)(H-z)^{-1}\chi_\ybf\|$ with respect to the distance of $\xbf$ and $\ybf$, where $\zeta$ is a suitably chosen smooth cutoff function which vanishes in a neighborhood of $\Real z$. This, however, is not sufficient for the purposes of this paper, as we need an exponentially decaying bound. We circumvent this problem by choosing a cutoff function that depends on $\xbf$ and $\ybf$, but nevertheless satisfies adequate estimates.

In the following, we assume that $H$ is a self-adjoint operator $H_\Omega$ on $L^2(\Omega)$ ($\Omega\subseteq\RD$ open) that satisfies a Combes-Thomas bound of the following form with respect to a pseudo-metric $\md$ on $\RD$:
\begin{equation}
\|\chi_\xbf(H_\Omega-z)^{-1}\chi_\ybf\| \leq \frac{C_0}{\dist(z,\sigma(H_\Omega))}\exp\left(\frac{-\mu_0\dist(z,\sigma(H_\Omega))}{1+|z|+\dist(z,\sigma(H_\Omega))}\md(\xbf,\ybf)\right)  \label{combesthomas}
\end{equation}
for all $\xbf,\ybf\in\RD$ and all $z\in\C\setminus\sigma(H_\Omega)$. Estimates of this type hold for a very general class of Schr\"odinger operators (cf. \cite{GK03}). In particular, we note that the random Schr\"odinger operator $H_\Omega^{(n)}(\omega)$ as defined in Section \ref{section:model} satisfies such an estimate with constants $C_0,\mu_0>0$ that are independent of the realization $V^{(n)}(\omega)$ of the random potential.

\begin{theorem}\label{rescutoffdecay}
Let $J\subset\R$ be a bounded interval and $r>0$. For any $\xbf,\ybf\in\RD$, there exists a smooth function $\chi:\R\rightarrow[0,1]$ such that $\chi(E)=0$ if $\dist(E,J)\leq r$, $\chi(E)=1$ if $\dist(E,J)\geq 2r$ and such that for all $z\in\C\backslash\R$ with $\Real z\in J$ the inequality
\begin{equation}
\bigg\|\chi_\xbf\frac{\chi(H_\Omega)}{H_\Omega-z}\chi_\ybf\bigg\| \leq Ce^{-\mu\md(\xbf,\ybf)} \label{result2} \ ,
\end{equation}
holds with positive constants $C$ and $\mu$ that depend only on $J$, $r$ and the constants in the Combes-Thomas estimate \eqref{combesthomas}.
\end{theorem}

As mentioned before, the choice of the cutoff function $\chi$ depends on the distance of $\xbf$ and $\ybf$. In the following, we present a lemma that specifies the conditions the cutoff function (multiplied by $(\cdot-z)^{-1}$) needs to satisfy. Here and in the following, we define $\langle E\rangle:=\sqrt{1+E^2}$ for any $E\in\R$.

\begin{lemma}\label{lemma:l1gevreyuptoN}
Suppose $I\supset\sigma(H_\Omega)$ is an open interval, $N\in\N_{\geq2}$ and $f\in C^\infty(I)$ such that
\begin{equation}
\int_I |f^{(k)}(E)|\langle E\rangle^{k-1}\dE \leq e^NC_{f,I}(C_{f,I}(k+1))^k \label{ineq:l1gevreyuptoN}
\end{equation}
for some $C_{f,I}\geq 1$ and all $k\in\{0,\ldots,N\}$. There exist $C_I,c_I>0$ depending only on $I$ and on the constants in the Combes-Thomas estimate such that
\begin{equation}
\|\chi_\xbf f(H_\Omega) \chi_\ybf\| \leq C_IC_{f,I}^3(\delta\md(\xbf,\ybf)+1)e^{-\delta\md(\xbf,\ybf)}
\end{equation}
for all $\xbf,\ybf\in\RD$ satisfying $\delta\md(\xbf,\ybf)\in[N-2,N-1)$, where $\delta=c_I/(C_{f,I}\ln(e^2C_{f,I}))$.
\end{lemma}

The assumption \eqref{ineq:l1gevreyuptoN} differs from the definition of \textit{$L^1$-Gevrey functions of class~$1$} in \cite{BGK} in two ways: We require the bound to hold only for $k\in\{1,\ldots,N\}$ and we allow an additional factor of $e^N$. As a consequence, the bound on the norm of $\chi_\xbf(H-z)^{-1}\chi_\ybf$ holds only in case $\md(\xbf,\ybf)$ lies in an $N$-dependent interval. However, the proof, which we postpone to the end of this section, proceeds like the analysis in \cite{BGK} and needs only a minor modifcation.

The lemma above is of use only if for any $N$, there is a cutoff function $\chi$ for which $\chi(\cdot)(\cdot-z)^{-1}$ satisfies \eqref{ineq:l1gevreyuptoN} with a uniform constant $C=C_{f,I}$. The existence of such a family of functions is ensured by the following proposition, which is a direct consequence of \cite{Ro}*{Proposition 1.4.10}.

\begin{proposition}\label{cutoffbound}
Let $N\in\N$, $r>0$ and $I_0\in\R$ be a bounded interval. Then there exists a function $\chi\in C^\infty(\R)$ with $0\leq \chi\leq 1$, $\chi(E)=1$ if $\dist(E,I_0)\geq r$ and $\chi(E)=0$ if $E\in I_0$ such that for all $E\in\R$ and $k\in\N_0$ with $k\leq N$
\begin{equation}
|\chi^{(k)}(E)|\leq c\bigg(\frac{AN}{r}\bigg)^k \ , \label{ineq:cutoffbound}
\end{equation}
where the constants $c,A>0$ are independent of $k$, $N$, $I_0$ and $r$.
\end{proposition}

We can now combine Lemma \ref{lemma:l1gevreyuptoN} and Proposition \ref{cutoffbound} in order to prove Theorem \ref{rescutoffdecay}.

\begin{proof}[Proof of Theorem \ref{rescutoffdecay}]
Let $\xbf,\ybf\in\RD$ be given. We define three auxiliary constants
\begin{align}
a_1&=\sup_{E\in\R\setminus I_0,\Real z\in J}\langle E\rangle/|E-z|, \notag\\ a_2&=\sup_{\substack{E\in\R\setminus I_0 \\ \dist(E,I_0)\leq r}}\langle E\rangle\text{ and }\notag\\
a_3&=\max\{2,Aa_2/(r a_1)\},
\end{align}
where $I_0:=[\inf J-r, \sup J+r]$ and $A$ is chosen as in \eqref{ineq:cutoffbound}. Now let $\delta:=c_I/(C_{f,I}\ln(e^2C_{f,I}))$ as in Lemma \ref{lemma:l1gevreyuptoN} with $C_{f,I}:=\max\{2\pi ca_1a_3,1\}$ ($c$ as in \eqref{ineq:cutoffbound}) and $I=\R$ and let $N\in\N$ such that $\delta\md(\xbf,\ybf)\in[N-2,N-1)$. Choose a cutoff function $\chi$ according to Proposition \ref{cutoffbound}. The function $\chi$ thus satisfies the bound
\begin{equation}
|\chi^{(k)}(E)|\leq c\bigg(\frac{AN}{r}\bigg)^k \leq c\bigg(\frac{A}{r}\bigg)^k e^N k!  \label{cutoffineq}
\end{equation}
for all $E\in\R$ and $k\leq N$.
We define $r_z(E):=(E-z)^{-1}$ and $f(E):=\chi(E)r_z(E)$, and we conclude
\begin{equation}
r_z^{(k)}(E)=(-1)^k k! \frac{1}{(E-z)^{k+1}}
\end{equation}
and hence
\begin{equation}
f^{(k)}(E) = \sum_{j=0}^k \binom{k}{j} \chi_J^{(j)}(E) \frac{(-1)^{k-j}(k-j)!}{(E-z)^{k-j+1}}.
\end{equation}
As a consequence, for any $k\in\{0,\ldots,N\}$, we have
\begin{align}
\int_\R |f^{(k)}(E)|\langle E\rangle^{k-1}\dE &\leq ck!\int_{\supp\chi} \frac{\langle E\rangle^{k-1}}{|E-z|^{k+1}}\dE \notag\\
&\quad +ce^Nk!\sum_{j=1}^k\left(\frac{A}{r}\right)^j\int_{\supp\chi^{(j)}}\frac{\langle E\rangle^{k-1}}{|E-z|^{k-j+1}}\dE \notag\\
&\leq ck!a_1^{k+1}\int_\R \langle E\rangle^{-2}\dE \notag\\
&\quad +ce^Nk!\sum_{j=1}^k\left(\frac{A}{r}\right)^ja_1^{k-j+1}a_2^j\int_{\supp\chi'}\langle E\rangle^{-2}\dE \notag\\
&\leq \pi ck!a_1^{k+1} + \pi ce^Nk!a_1^{k+1}\sum_{j=1}^k\left(\frac{Aa_2}{r a_1}\right)^j \notag\\
&\leq 2\pi ce^Nk!a_1^{k+1}a_3^{k+1} \notag\\
&\leq e^NC_{f,I}(C_{f,I}(k+1))^k.
\end{align}
Lemma \ref{lemma:l1gevreyuptoN} now implies
\begin{equation}
\left\|\chi_\xbf\frac{\chi(H_\Omega)}{H_\Omega-z}\chi_\ybf\right\| = \|\chi_\xbf f(H_\Omega) \chi_\ybf\| \leq C_IC_{f,I}^3(\delta\md(\xbf,\ybf)+1)e^{-\delta\md(\xbf,\ybf)} ,
\end{equation}
where $C_I$, $C_{f,I}$ and $\delta$ depend only on $J$, $r$ and the constants in the Combes-Thomas estimate. This implies the claim of the theorem.
\end{proof}
It remains to prove Lemma \ref{lemma:l1gevreyuptoN}:

\begin{proof}[Proof of Lemma \ref{lemma:l1gevreyuptoN}]
The analysis in \cite{BGK} shows that, up to a factor depending only on $I$ and the constants in the Combes-Thomas estimate, the norm of $\chi_\xbf f(H_\Omega)\chi_\ybf$ can be bounded by the sum of three terms:
\begin{align}
&e^Nn(eC_{f,I})^{n+2}e^{n\ln\epsilon},\notag\\
&e^Nn(eC_{f,I})^{n+2}e^{-c\md(\xbf,\ybf)},\notag\\
&e^Nn(eC_{f,I})^{n+2}e^{-c\epsilon\md(\xbf,\ybf)}.
\end{align}
Here $c>0$ depends only on $I$ and the constants in the Combes-Thomas estimate and both $n$ and $\epsilon$ can be chosen arbitrarily, as long as they satisfy $n\in\{1,\ldots,N-1\}$ and $0<\epsilon\leq 1/2$. Now suppose $\delta\md(\xbf,\ybf)\in[N-2,N-1)$ with $\delta=c/(2e^3C_{f,I}\ln(e^2C_{f,I}))$. Then, choosing $n=N-1$ and $\epsilon=1/(e^3C_{f,I})$, we can estimate
\begin{align}
e^Nn(eC_{f,I})^{n+2}e^{n\ln\epsilon} &\leq (N-1)eC_{f,I}e^{N+N\ln(eC_{f,I})-(N-1)\ln(e^3C_{f,I})} \notag\\
&\leq (\delta\md(\xbf,\ybf)+1)e^5C_{f,I}^3e^{\delta\md(\xbf,\ybf)(1+\ln(eC_{f,I})-\ln(e^3C_{f,I}))} \notag\\
&=(\delta\md(\xbf,\ybf)+1)e^5C_{f,I}^3\exp\left(-\frac{c}{2e^3C_{f,I}\ln(e^2C_{f,I})}\md(\xbf,\ybf)\right)
\end{align}
and
\begin{align}
e^Nn(eC_{f,I})^{n+2}e^{-c\md(\xbf,\ybf)} &\leq e^Nn(eC_{f,I})^{n+2}e^{-c\epsilon\md(\xbf,\ybf)} \notag\\
&\leq (\delta\md(\xbf,\ybf)+1)e^5C_{f,I}^3e^{|x-y|(\delta(1+\ln(eC_{f,I}))-c\epsilon)} \notag\\
&= (\delta\md(\xbf,\ybf)+1)e^5C_{f,I}^3\exp\left(-\frac{c}{2e^3C_{f,I}}\md(\xbf,\ybf)\right).
\end{align}
As $\ln(e^2C_{f,I})\geq 1$, this implies the claimed bound.
\end{proof}

\section*{Acknowledgements}
Michael Fauser was partially supported by TopMath, the graduate program of the Elite Network of Bavaria and the graduate center of TUM Graduate School.

\end{document}